%% file: main.tex
\documentclass[11pt]{article}

\input{00_preamble}

\title{\bfseries
    A Limit Order Market with Uncertain Informed Trading Participation
}

\author{
    Umut \c{C}etin\thanks{\href{mailto:u.cetin@lse.ac.uk}{u.cetin@lse.ac.uk}}
    \quad
    Mingwei Lin\thanks{\href{mailto:m.lin20@lse.ac.uk}{m.lin20@lse.ac.uk}}
}

\affil{
    Department of Statistics, London School of Economics
}

\begin{document}
\date{}  
\maketitle

\begin{abstract}
We study a one period limit order market with informed traders, noise traders, and competitive liquidity suppliers, in which the number of informed traders is random. Liquidity suppliers know the distribution of the informed trader count, but not its realization, and therefore face uncertainty about both the presence and the intensity of informed trading. We characterize equilibrium by a fixed point integral equation for the marginal cost function and establish existence of equilibrium for bounded asset values. We then analyse large order asymptotics. For bounded asset values with power law endpoint behaviour, equilibrium price impact follows a power law whose exponent is determined jointly by the asset value tail and the full distribution of the informed trader count. In particular, this exponent is not determined by the expected number of informed traders alone. In the light endpoint regime, price impact is instead logarithmic. Finally, we solve the fixed point numerically across several asset value and informed trader count distributions. The numerical results are consistent with the theoretical asymptotics in the cases covered by the theory and provide comparative statics beyond them.
\end{abstract}

\section{Introduction}

A central question in market microstructure is how adverse selection shapes the price response to order flow. 
Classical asymmetric information models formalize this mechanism in settings where the informational structure is known to market participants. 
In the auction style setting of \cite{Kyle_1985}, a single risk neutral informed trader trades against noise traders and a risk neutral market maker, and equilibrium price impact is linear. The continuous time extension of \cite{Back1992} preserves this linear price impact structure. In quote driven and limit order book (LOB) markets, however, trades are not all executed at a single market clearing price. Transaction prices vary with order size, direction, and the state of the book. 
In the sequential trade model of \cite{GlostenMilgrom_1985}, the positive bid ask spread reflects adverse selection and depends on the probability of informed trading. A large body of LOB models \citep[e.g.,][]{Glosten_1994,parlour.1998, foucault.1999, rosu.2009, cont_etal.2014} further studies discriminatory pricing, order placement, adverse selection, book depth, and nonlinear price impact.

Price impact provides a direct measure of how order flow is incorporated into prices and how liquidity adjusts to adverse selection. Empirical evidence suggests that market impact is nonlinear and concave in trade size, with functional forms ranging from square root to logarithmic \citep[e.g.,][]{Barra1997,potters_bouchaud.2003, almgren.2005, Moro_etal.2009, bershova_andrakhlin.2013, zarinelli_etal.2015, toth_etal.2016}. While single price auction models are relatively tractable, LOB models require the equilibrium determination of an entire pricing rule, namely a function mapping order size into marginal transaction prices, rather than a scalar price.
In this paper, we study this problem in a rational expectations limit order market with uncertainty about informed trading participation.

A common assumption in theoretical models of asymmetric information is that the presence, or number, of informed traders is common knowledge.
This assumption underlies many classical analyses of price formation, market efficiency, and equilibrium structure.
For example, \cite{Grossman_1980} studies how the number of informed traders affects the informativeness of prices when the proportion of informed traders is common knowledge in a rational expectations equilibrium. Similarly, each investor in \cite{hellwig.1980} is certain about the number of other investors. Such assumptions impose a strong degree of knowledge about the informational structure of the market.

In practice, market participants are uncertain not only about asset fundamentals, but also about whether informed traders are present and how intensively they trade. \cite{easley_ohara.1987} propose a sequential trade model in which market makers face uncertainty about whether trades come from informed or uninformed investors, and update their beliefs using trade direction and discrete trade size. This leads to the probability of informed trading (PIN) as a measure of adverse selection. \cite{easley_etal.2014} consider a different form of informational uncertainty in a Gaussian framework, where market participants are uncertain not about the existence of informed traders, but about their strategic behaviour, a setting they call opaque trading. \cite{banerjee_green.2015} study a market in which uninformed traders are uncertain whether they are trading against rational informed traders or irrational noise traders, with only one type active at a time. \cite{avery_zemsky.1998} and \cite{papadimitriou.2023} consider models in which the proportion of informed traders is uncertain and learned over time. \cite{gao_etal.2013} introduces uncertainty about the proportion of informed traders in a rational expectations equilibrium and obtains nonlinear price impact. \cite{li.2013} and \cite{back_etal.2013} extend \cite{Kyle_1985} and \cite{Back1992} to continuous time settings with uncertainty about the existence of a monopolistic informed trader.
Recent empirical work also demonstrates a growing interest in quantifying informed trading uncertainty. \cite{collin-dufresne_fos.2015} provide evidence that informed traders strategically time their participation, especially around major information events. \cite{bogousslavsky_etal.2024} develop machine learning methods, including gradient boosted trees, to estimate informed trading intensity from high frequency data and improve the detection of informational asymmetries.

This paper introduces random informed trader participation into a rational expectations limit order market and studies its effect on equilibrium liquidity and large order market impact.
Liquidity suppliers observe only aggregate order flow. They know the distribution of the informed trader count, but not its realization. 
Hence they face uncertainty not only about the asset value, but also about the intensity of informed trading. 
As a result, the equilibrium pricing rule depends on the entire distribution of informed participation.
The main result is that the law of the informed trader count affects not only the level of liquidity, but also the asymptotic form of equilibrium price impact.

The equilibrium is characterized by a fixed point integral equation for the marginal cost function \(F\). Given \(F\), the equilibrium limit order book is recovered from the pricing operator \(h^*=\phi_F\). 
This formulation accommodates general asset value distributions and general distributions for the number of informed traders. 
In the uncertain monopolistic benchmark, the model yields explicit expressions for the bid ask spread and the informed trader's expected profit, illustrating how uncertainty about the presence of an insider affects the bid--ask spread.

The second contribution concerns the large order asymptotics of equilibrium price impact. For bounded asset values with power law endpoint behaviour, equilibrium price impact follows a power law. The corresponding exponent is characterized by a fixed point equation that depends jointly on the endpoint behaviour of the asset value and the distribution of the number of informed traders. Thus, the distribution of informed trading participation affects the asymptotic shape of the limit order book. In particular, the effect is not summarized by the expected number of informed traders alone. In the light endpoint regime, price impact is instead logarithmic.

The final contribution is numerical. The numerical experiments are consistent with the predicted power law and logarithmic tail behaviour in the cases covered by the theory, and provide exploratory evidence for specifications outside the compact support assumptions. We also study comparative statics with respect to the distribution of the informed trader count and the bid ask spread under a two point count distribution. In the specifications considered, greater dispersion in the informed trader count, holding the conditional mean fixed, flattens the book.

The remainder of the paper is organized as follows. Section~\ref{section:MarketStructure} introduces the market structure. Section~\ref{section:Equilibrium} characterizes equilibrium and establishes its existence. Within this section, Section~\ref{section::monopolistic} studies the uncertain monopolistic informed trader benchmark, while Section~\ref{section::generalRandom} develops the general model with a random number of informed traders. Section~\ref{section:PriceAsymptotics} analyses the large order asymptotics of equilibrium price impact. Section~\ref{section:Numerics} presents the numerical experiments, and Section~\ref{section:Conclusion} concludes.

\section{The Structure of the Limit Order Market}\label{section:MarketStructure}
We consider a one period limit order market with competitive liquidity suppliers, noise traders, and a random number of informed traders. The model follows the limit order market structure of~\cite{Cetin_Waelbroeck_2023}, but introduces uncertainty about informed trading participation. Liquidity suppliers do not observe whether informed traders are present, nor how many of them are active. They know only the distribution of the informed trader count. As a result, the limit order book reflects not only adverse selection about the asset value, but also uncertainty about the intensity of informed trading.

All random variables discussed in this section are presumed to be defined on a complete probability space $(\Omega,\mathcal{F},\mathcal{P})$, and $E$ denotes the expectation operator corresponding to $\mathcal{P}$.

In this one-period limit order market, a single asset is traded through the limit order book. The trades are assumed to take place at $t=0$ and $t=1$. The risk-free interest rate is set to be $r=0$. The fundamental value of the trading asset $V$ will be revealed to the public at time $t=1$. The limit order market consists of four types of agents:
\begin{enumerate}
    \item Competitive (infinitely many) liquidity suppliers. They do not know whether there are informed traders in the market or how many there are. Their only shared knowledge concerning the insiders is the distribution of the quantity of informed traders within the market. 
    Based on their information set, they move first and place limit orders to the limit order book. The resulting limit order book is characterised by a nonconstant, nondecreasing function $h:\mathbb{R}\rightarrow\mathbb{R}$. Therefore, a market order of size $x$ traded against the limit order book incurs the cost of 
    \begin{align*}
        \int_{0}^{x}h(y)\, dy.
    \end{align*}
    The function $h(y)$ represents the marginal price of the $y$-th share. 
    \item Noise traders. Their total trading demand is given by $ Z \sim N(0,\sigma^2)$. They are non-strategic, i.e. $Z$ is independent of the asset fundamental value $V$.
    \item Risk neutral (symmetric) informed traders $N$. The number of informed traders $N$ is a discrete random variable with probability mass function $p(n):= P(N=n),$ with $n=0,1,\dots,\, E[N]<\infty$. Moreover, the quantity of informed traders in the market is exogenous. That is, $N,\,V$ and $Z$ are mutually independent. When present, informed traders observe the terminal value $V$ at time $t=0$ and choose the (same) trade size $x$ to maximise their expected profit conditional on their private information. 
    \item A trading desk. The trading desk acts as a broker, it only gathers and executes orders for its clients. Specifically, the trading desk aggregates all orders and trades $Z+Nx$ with the liquidity suppliers.
\end{enumerate}
Moreover, noise traders orders are assumed to arrive at the desk simultaneously and prior to the informed orders, when present, and the orders are charged proportional to their order size. That is, if $U$ denotes the aggregate demand of the other informed traders, the cost of trading $x$ units for an individual informed trader is 
\begin{align}\label{eq:cost_structure}
    \frac{x}{U+x} \int_{0}^{U+x}h(Z+y)\, dy.
\end{align}

\section{Insiders' Optimal Strategies and Limit Price in Equilibrium}\label{section:Equilibrium}
In this section, we first consider the case of a monopolistic informed trader in the market, while the liquidity suppliers are unsure about her presence. We then extend the analysis to a general random informed traders.
\subsection{Uncertain Monopolistic Informed Trading}\label{section::monopolistic}
When there is a monopolistic informed trader who knows the terminal fundamental value of the asset $V$ at time $t=0$, she seeks to exploit her information advantage by determining the trade size $x$ to maximise her expected profit. However, liquidity suppliers in the market are uncertain about the existence of such an informed trader. As a result, the monopolistic insider can obtain higher expected profits within the limit order market. Meanwhile, liquidity suppliers shall adjust the bid-ask spread, and consequently, their expected revenues from noise traders, according to their beliefs on the presence of the insider.
\subsubsection{Optimal strategy for the monopolistic informed trader}
The optimal trading size $x$ for the monopolistic informed trader is
\begin{align*}
    \arg\max_{x\in\mathbb{R}}E\left[Vx-\int_{0}^{x}h(Z+u)\,du\Big\lvert V=v\right].
\end{align*}
Since $h$ is non-decreasing, first-order condition is sufficient to this optimisation problem. Define
\begin{align*}
    F(x):=E\left[h(Z+x)\right]=\int_{-\infty}^{\infty}q(\sigma,z-x)h(z)\,dz,
\end{align*}
which is the expected marginal cost of the $x$-th share. Then, given a limit order book $h$, the optimal strategy for the monopolistic insider is to trade $x^*$ such that $F(x^{*})=V$. Moreover, $F(x)$ is strictly increasing and one-to-one, since $h(\cdot)$ is not constant and non-decreasing.
\subsubsection{The limit price and the equilibrium}
The limit prices following \cite{Glosten_1994} are given by tail expectations. That is, 
\begin{align}\label{eq:limit_price}
    h(y):=
    \begin{cases}
        E[V\lvert Y\geq y] , \text{ if $ y>0 $,}\\
        E[V\lvert Y< y] , \text{ if $ y<0 $},\\
    \end{cases}
\end{align}
where $Y$ is the total demand from insider, when present, and noise traders. These limit prices guarantee the zero expected aggregate profit for liquidity suppliers in the sense that
\begin{align*}
    & E\Big[\int_{0}^{Y}(h(y)-V)\, dy\Big] \\
    = &\int_{0}^{\infty} E[(h(y)-V)1_{[Y\geq y]}] \, dy + \int_{-\infty}^{0} E[(V-h(y))1_{[Y\leq y]}] \, dy = 0.
\end{align*}
Moreover, the bid-ask spread is characterised by $h(0+)-h(0-):=\lim_{y\rightarrow 0^{+}}h(y)-\lim_{y\rightarrow 0^{-}}h(y)$.

Liquidity suppliers are uncertain about whether there is an informed trader in the market. Their common beliefs are that the probability of having an insider in the market is $p$. In other words, $N\sim \text{Bernoulli}(p),\,0\leq p\leq 1$. Therefore, for $y>0$, the marginal ask price is
\begin{align*}
    h(y) &= \frac{E\left[V 1_{[Y \geq y]}\right]}{P(Y \geq y)} = \frac{E\left[V 1_{ [Nx + Z \geq y]}\right]}{P(Nx + Z \geq y)}\\
         &= \frac{(1-p) E\left[V 1_ {[Z \geq y]}\right] + p E\left[V 1_{ [x + Z \geq y]}\right]}{(1-p) P(Z \geq y) + p P(x + Z \geq y)}.
\end{align*}
In equilibrium, we have $x^{*}=F^{-1}(V)$. And denote the limit prices in equilibrium by $h^{*}(y)$, since $V\perp Z$,
\begin{align*}
    h^{*}(y)&=\frac{(1-p)E[V]P(Z\geq y)+p E[V1_{[x^{*}+Z\geq y]}]}{(1-p)P(Z\geq y)+p P(x^{*}+Z\geq y)}\\
    &=\frac{(1-p)E[V]P(Z\geq y)+pE[V1_{[V\geq F(y-Z)]}]}{(1-p)P(Z\geq y)+pP(V\geq F(y-Z))}.
\end{align*}
\begin{example}\label{example:binary}
Consider a binary asset value satisfying \(P(V=1)=P(V=-1)=1/2\), as in Example~3.4 of \cite{Cetin_Waelbroeck_2023}. This example shows how uncertainty about the existence of a monopolistic informed trader affects both the bid--ask spread and the insider's expected profit.

By monotonicity of \(F\), \(\lim_{x\rightarrow\infty}F(x)=1\) and \(\lim_{x\rightarrow-\infty}F(x)=-1\). Hence the insider trades \(x^*=+\infty\) when \(V=1\), and \(x^*=-\infty\) when \(V=-1\). Suppose liquidity suppliers assign probability \(p\in(0,1]\) to the existence of the insider. For \(y>0\), the equilibrium ask price is
\[
\begin{aligned}
    h^*(y)
    &=
    \frac{pP(V=1)}
    {(1-p)P(Z\ge y)+pP(V=1)}
    =
    \frac{1}{1+\frac{2(1-p)}{p}P(Z\ge y)} .
\end{aligned}
\]
Similarly, for \(y<0\), the equilibrium bid price is
\[
    h^*(y)
    =
    -\frac{1}{1+\frac{2(1-p)}{p}P(Z<y)} .
\]
Thus \(h^*(0+)=p\), \(h^*(0-)=-p\), and the equilibrium bid--ask spread is \(h^*(0+)-h^*(0-)=2p\). When \(p=0\), there is no informed trader and \(h(y)=0\) for all \(y\). When \(p=1\), the limit price becomes \(h^*(y)=1_{\{y>0\}}-1_{\{y<0\}}\), recovering the deterministic monopolistic-insider case. Hence uncertainty about the insider reduces the spread relative to the deterministic case, while the spread increases with the perceived probability \(p\).

The same uncertainty increases the monopolistic insider's expected profit. For \(V=1\),
\[
\begin{aligned}
& E\left[
    \int_0^\infty \left(1-h^*(Z+y)\right)\,dy
    \,\bigg|\, V=1
\right] \\
\qquad =&
\sigma\sqrt{\frac{2}{\pi}}
+
\int_0^\infty
\left(
1-\frac{1}{1+\frac{2(1-p)}{p}P(Z\ge y)}
\right)
\left(P(Z>-y)-P(Z<-y)\right)\,dy \\
\qquad \geq&
\sigma\sqrt{\frac{2}{\pi}} .
\end{aligned}
\]
The correction term is non-negative and finite, since
\(0\leq 1-(1+\frac{2(1-p)}{p}P(Z\ge y))^{-1}
\leq \frac{2(1-p)}{p}P(Z\ge y)\), and
\(\int_0^\infty P(Z\ge y)\,dy<\infty\). Thus, compared with the deterministic monopolistic-insider case, uncertainty about the insider lowers the spread but increases the insider's expected profit.
\end{example}

The preceding example is special because the binary asset value leads to full
revelation whenever the insider is present. For a general asset distribution, the
bid--ask spread can be expressed in terms of the equilibrium marginal cost
function \(F\). Let \(x^*(V)=F^{-1}(V)\). If
\(Z\sim\mathcal N(0,\sigma^2)\), then
\[
    h^*(0+)
    =
    \frac{
        E\left[V\Phi\left(\frac{NF^{-1}(V)}{\sigma}\right)\right]
    }{
        E\left[\Phi\left(\frac{NF^{-1}(V)}{\sigma}\right)\right]
    },
    \qquad
    h^*(0-)
    =
    \frac{
        E\left[V\Phi\left(-\frac{NF^{-1}(V)}{\sigma}\right)\right]
    }{
        E\left[\Phi\left(-\frac{NF^{-1}(V)}{\sigma}\right)\right]
    } .
\]
Thus the spread is \(h^*(0+)-h^*(0-)\). In the two-point case
\(P(N=0)=1-p\) and \(P(N=m)=p\), this becomes
\begin{equation}\label{eq:twopoint_spread}
    h^*(0\pm)
    =
    \frac{
        \frac{1-p}{2}E[V]
        +pE\left[V\Phi\left(\pm\frac{mF^{-1}(V)}{\sigma}\right)\right]
    }{
        \frac{1-p}{2}
        +pE\left[\Phi\left(\pm\frac{mF^{-1}(V)}{\sigma}\right)\right]
    } .
\end{equation}
Here \(p\) is the probability of informed trading, while \(m\) is the number of
informed traders conditional on their presence. For a fixed \(F\), the formula
shows how increasing \(p\) or \(m\) changes the spread in the two-point
specification. This monotonicity, however, need not extend to a general
distribution of \(N\), since the law of \(N\) enters both the bid--ask spread
expression and the equilibrium fixed point determining \(F\).

From the example, the uncertainty about the insider raises the
insider's expected profit. Appendix~\ref{appendix:profit} shows numerically that the same pattern persists for continuous value distributions and a general random number of competing insiders.

\subsection{Random Numbers of Informed Traders}\label{section::generalRandom}
We now consider the case in which liquidity suppliers in the market are still uncertain about the existence of informed trader. Furthermore, the market may contain multiple informed traders. Assume that all market participants, including liquidity suppliers, informed traders, and noise traders, share a common belief regarding informed traders in the market, characterised by an exogenous random variable $N=0,1,2,\dots,\, E[N]<\infty$. Consequently, liquidity suppliers provide limit prices based on this understanding of uncertainty. Meanwhile, informed traders will exploit their private information advantage while considering the presence of other insiders.
\subsubsection{Optimal strategies for informed traders}
For an individual informed trader, the number of insiders in the market follows a conditional distribution, conditioned on her own existence. This conditional law represents an informed trader's belief about the number of competing insiders given her own presence. Specifically,
\begin{align*}
P(N=n\lvert N\geq 1)=
    \begin{cases} 
    0, &  n = 0, \\
    \frac{p(n)}{1-p(0)}, & n = 1, 2, 3, \ldots.
\end{cases}
\end{align*}
Given the cost structure in \eqref{eq:cost_structure}, an individual informed trader, conditional on being present, chooses $x$ to maximize the expected profit, by choosing a trade size $x$ conditioning on her private information on the asset value $V$ and on $N\geq 1$. That is,
\begin{align}
    \arg\max_{x\in\mathbb{R}}E\left[ Vx -  \frac{x}{U+x} \int_{0}^{U+x}h(Z+y)\, dy \Big\lvert V=v, N\geq1\right],
\end{align}
where $U$ is the random variable, denoting the total trading amount from the other informed traders.
The first order condition of the above maximization problem gives us
\begin{align*}
    V=E^{v}\left[ \frac{x}{U+x}h(Z+U+x)+\frac{U}{(U+x)^2}\int_{0}^{U+x}h(Z+y)\,dy \Big\lvert N\geq 1 \right],
\end{align*}
where $E^v$ is the expectation operator for the informed trader with the private information $V=v$. Assume every insider has symmetric information and is risk neutral, i.e. $U=(N-1)x$, the first order condition associate with the above optimisation problem of an individual insider is given by
\begin{align*}
    V = E^{v}\left[\frac{h(Z+Nx^{*})}{N}+\frac{N-1}{N^{2}x^{*}}\int_{0}^{Nx^{*}}h(Z+u)\,du\Big\lvert N\geq1\right].
\end{align*}
To characterize the optimal strategy for each of the insider, the first order condition can be written as $ V = F(x^{*}) $, where
\begin{align} \label{eq:F(x)}
    \begin{split}
        F(x):=&E^{v}\left[\frac{h(Z+Nx)}{N}+\frac{N-1}{N^{2}x}\int_{0}^{Nx}h(Z+u)\,du\Big\lvert N\geq1\right]\\
        =&\sum_{n\geq1}\frac{p(n)}{1-p(0)}\int_{-\infty}^{\infty}\left\{ \frac{1}{n}q(\sigma,z-nx)+\frac{n-1}{n}\Bar{q}(\sigma,n,x,z) \right\}h(z)\,dz,
    \end{split}
\end{align}
where $Z\sim N(0,\sigma^2),\,q(\sigma,\cdot)$ is the probability density function of a mean-zero Gaussian random variable with variance $\sigma^2$, and $\Bar{q}(\sigma,n,x,z):=1_{[x\neq0]}\frac{1}{x}\int_{0}^{x}q(\sigma,z-ny)\,dy + 1_{[x=0]}q(\sigma,z)$.

Since \(h\) is nonconstant and nondecreasing, convolution with the strictly positive Gaussian density implies that \(F\) is strictly increasing. Hence the optimal strategy is obtained by inverting \(F\), namely \(x^*=F^{-1}(V)\). 

\subsubsection{The limit price and the equilibrium}
The limit price is in the same form as \eqref{eq:limit_price}.
In the present setting, the total demand is $ Y=Nx+Z $, is equal to the total demand from insiders $ Nx $ and the total demand from noise traders $ Z $.
\begin{definition}\label{def:V_operators}
    Suppose the support of the random variable $ V $ is $ (m,M) $, where $ -\infty\leq m<M\leq\infty $. We define the following right-continuous functions:
    \begin{align*}
        \Phi^{+}(y)&:=E[V1_{[V>y]}],\, \Pi^{+}(y):=P(V>y),\\
        \Phi^{-}(y)&:=E[V1_{[V\leq y]}],\, \Pi^{-}(y):=P(V\leq y)=1-\Pi^{+}(y).
    \end{align*}
    Furthermore, define on the support of $ V $:
    \begin{align*}
        \Psi^{\pm}(y)&:=\frac{\Phi^{\pm}(y)}{\Pi^{\pm}(y)},
    \end{align*}
    so that $ \Psi^{+}(y) =E[V\lvert V>y]$ and $ \Psi^{-}(y):=E[V\lvert V\leq y] $. Note that $\Phi^{+}(x-)=E[V1_{[V\geq x]}]=\Phi^{+}(x)$ for almost all $x$. Analogously for $\Pi^{+}(x-)$ and $\Psi^{+}(x-)$.
\end{definition}

Let us consider the ask-side limit price, that is, when $y>0$,
\begin{align*}
    h(y)
    &=E[V\lvert Y\geq y]=\frac{E[V1_{[Y\geq y]}]}{P(Y\geq y)}=\frac{E[V1_{[Nx+Z\geq y]}]}{P(Nx+Z\geq y)}\\
    &=\frac{p(0)E[V1_{[Z\geq y]}]+\sum_{n\geq0}p(n)1_{[n\neq0]}E[V1_{[nx+Z\geq y]}]}{p(0)P(Z\geq y)+\sum_{n\geq0}p(n) 1_{[n\neq 0]}P(nx+Z\geq y)}. 
\end{align*}
In equilibrium, $F(x^{*})=V$, which gives us $x^{*}=F^{-1}(V)$. The limit price in equilibrium is therefore
\begin{align}
    h^{*}(y)
    &=\frac{p(0)E[V1_{[Z\geq y]}]+\sum_{n\geq0}p(n)1_{[n\neq0]}E[V1_{[V\geq F(\frac{y-Z}{n})]}]}{p(0)P(Z\geq y)+\sum_{n\geq0}p(n) 1_{[n\neq 0]}P(V\geq F(\frac{y-Z}{n}))}\nonumber\\
    &=\frac{E[V1_{[N=0]}1_{[Z\geq y]}+1_{[N\neq0]}\Phi^{+}(F(\frac{y-Z}{N}))]}{E[1_{[N=0]}1_{[Z\geq y]}+1_{[N\neq0]}\Pi^{+}(F(\frac{y-Z}{N}))]}\label{eq:h^{*}(y)_def}\\
    &=\frac{E[V1_{[N=0]}1_{[B_1 \leq 0]}+1_{[N\neq0]}\Phi^{+}(F(\frac{B_1}{N}))]}{E[1_{[N=0]}1_{[B_1 \leq 0]}+1_{[N\neq0]}\Pi^{+}(F(\frac{B_1}{N}))]},\label{eq:h^{*}(y)_equiv_pf}
\end{align}
where $B_{t}=y+\sigma\beta_t$, $\beta_t$ is a standard Brownian motion. And $B_1$ is this Brownian motion at $t=1$. Similarly for the bid side prices $h^{*}(y)$, when $y<0$.\\
In order to obtain an equation of $F$ in equilibrium, it is convenient to define the mappings associated with $F$.
\begin{definition}\label{def:phi_g}
    For any continuous function $g$, according to \eqref{eq:h^{*}(y)_def}, define the mappings
    \begin{align*}
        \phi^{+}_{g}(x):=& \frac{E[V1_{[N=0]}1_{[Z\geq x]}+1_{[N\neq0]}\Phi^{+}(g(\frac{x-Z}{N}))]}{E[1_{[N=0]}1_{[Z\geq x]}+1_{[N\neq0]}\Pi^{+}(g(\frac{x-Z}{N}))]}\\
        \phi^{-}_{g}(x):=&\frac{E[V1_{[N=0]}1_{[Z\leq x]}+1_{[N\neq0]}\Phi^{-}(g(\frac{x-Z}{N}))]}{E[1_{[N=0]}1_{[Z\leq x]}+1_{[N\neq0]}\Pi^{-}(g(\frac{x-Z}{N}))]}\\
        \phi_{g}(x):=&\phi_{g}^{+}(x)1_{[x>0]}+\phi_{g}^{-}(x)1_{[x<0]}.
    \end{align*}
\end{definition}
Thus, the limit prices in equilibrium can be written as $h^{*}(x)=\phi_{F}(x),\,x\in\mathbb{R}$. Combining this expression with~\eqref{eq:F(x)} yields the equation of $F$ in equilibrium.
\begin{align}\label{eq:F_FixedPointMapping}
    F(x)=\sum_{n\geq1}\frac{p(n)}{1-p(0)}\int_{-\infty}^{\infty}\left\{ \frac{1}{n}q(\sigma,z-nx)+\frac{n-1}{n}\bar{q}(\sigma,n,x,z) \right\}\phi_{F}(z)\,dz, 
\end{align}
where
\begin{align}
    \bar{q}(\sigma,n,x,z):=1_{[x\neq0]}\frac{1}{x}\int_{0}^{x}q(\sigma,z-ny)\,dy + 1_{[x=0]}q(\sigma,z). 
\end{align}
The preceding derivations show that the existence of equilibrium in the limit order market is equivalent to find $F$ so that \eqref{eq:F_FixedPointMapping} is satisfied. In other words, finding the equilibrium reduces to solving the fixed point mapping problem \eqref{eq:F_FixedPointMapping}.
\begin{theorem}
    Equilibrium exists if and only if there exists a function $F: \mathbb{R}\rightarrow\mathbb{R}$ that satisfies \eqref{eq:F_FixedPointMapping}. Given such a solution $F$, $(x^{*},\phi_{F})$ constitutes an equilibrium, where $x^{*}=F^{-1}(V)$ is the optimal strategy for an individual insider, and $\phi_{F}$ is the limit prices in equilibrium as in Definition \ref{def:phi_g} and \eqref{eq:h^{*}(y)_def}.
\end{theorem}
\subsection{Existence of the equilibrium}

We shall show the existence of an equilibrium when the asset value $V$ has bounded support $[m,M]$. Throughout this subsection, we impose a mild lower-tail finiteness condition. This condition is used to justify the interchange of integrals and limits involving the lower tail of the Gaussian kernel in the arguments below.

\begin{assumption}\label{assumption:lower_tail_finiteness}
For any function \(F\) considered below, the corresponding limit price function \(\phi_F\) satisfies
\[
    \int_{-\infty}^{0}\phi_F(z)q(\sigma,z)\,dz>-\infty .
\]
\end{assumption}

We first collect the main regularity and monotonicity properties of the pricing operator $g\rightarrow \phi_g$ introduced in Definition \ref{def:phi_g}. These properties will be used to study non-decreasing solutions of the fixed point equation.

\begin{lemma}\label{lemma:phi_property}
    Let \(g:\mathbb{R}\rightarrow[m,M]\) be a continuous function and \(\Psi^+(M-)=M\) and \(\Psi^-(m+)=m\). Suppose \(u^{+}\) (resp. \(u^{-}\)) is the unique solution of
    \begin{align*}
        u_{t}+\frac{1}{2}\sigma^2 u_{xx} &= 0,\\
        u(1,x)&=E_{N}\left[ 1_{[N=0]}1_{[x\leq0]}+1_{[N\neq0]}\Pi^{+}\left(g\left(\frac{x}{N}\right)\right) \right],\\
        (resp.\, u(1,x)&=E_{N}\left[ 1_{[N=0]}1_{[x\geq0]}+1_{[N\neq0]}\Pi^{-}\left(g\left(\frac{x}{N}\right)\right) \right]),
    \end{align*}
    where \(N\) is random with probability mass function \(p(n):=P(N=n)\), \(n=0,1,\dots\), and \(E_N\) in the terminal condition is the expectation with respect to \(N\). Then the following statements hold:
    \begin{itemize}
        \item[(1)] There exists a probability measure \(\mathbb Q\) on the joint space of \((B,N)\) such that the process \(B\) solves
        \begin{align}\label{eq:SDE}
            dB_{t}
            =
            \sigma dW_{t}^{\mathbb Q}
            +
            \sigma^{2}\frac{u_{x}(t,B_t)}{u(t,B_{t})}\,dt,
            \qquad B_{0}=x,
        \end{align}
        where \(u\) is either \(u^{+}\) or \(u^{-}\), and \(W_t^{\mathbb Q}\) is a Brownian motion under \(\mathbb Q\) with \(W_{0}=0\).

        \item[(2)] Suppose \((B,\mathbb Q^{+})\) and \((B,\mathbb Q^{-})\) correspond to the solutions of \eqref{eq:SDE} when \(u=u^{+}\) and \(u=u^{-}\), respectively, and \(E^{\mathbb Q}\) denotes expectation under \(\mathbb Q\). Then
        \[
            \phi_{g}^{+}(x)
            =
            E^{\mathbb Q^{+}}\left[
            E[V]1_{[N=0]}1_{[B_{1}\leq0]}
            +
            1_{[N\neq0]}\Psi^{+}\left(g\left(\frac{B_1}{N}\right)\right)
            \right],
        \]
        and
        \[
            \phi_{g}^{-}(x)
            =
            E^{\mathbb Q^{-}}\left[
            E[V]1_{[N=0]}1_{[B_{1}\geq0]}
            +
            1_{[N\neq0]}\Psi^{-}\left(g\left(\frac{B_1}{N}\right)\right)
            \right].
        \]

        \item[(3)] \(\phi_{g}^{+}(0)>\phi_{g}^{-}(0)\).

        \item[(4)] Suppose further that \(g\) is nondecreasing. Then \(\phi_{g}^{+}\) and \(\phi_{g}^{-}\) are nondecreasing. Consequently, \(\phi_g\) is nondecreasing. Moreover,
        \begin{align*}
            \phi_{g}^{+}(x)
            &\leq
            E^{\mathbb Q^{+}}\left[
            E[V]1_{[N=0]}1_{[\sigma W_{1}^{\mathbb Q^{+}}+x\leq0]}
            +
            1_{[N\neq0]}\Psi^{+}\left(
            g\left(
            \frac{\sigma W_{1}^{\mathbb Q^{+}}+x}{N}
            \right)
            \right)
            \right],\\
            \phi_{g}^{-}(x)
            &\geq
            E^{\mathbb Q^{-}}\left[
            E[V]1_{[N=0]}1_{[\sigma W_{1}^{\mathbb Q^{-}}+x\geq0]}
            +
            1_{[N\neq0]}\Psi^{-}\left(
            g\left(
            \frac{\sigma W_{1}^{\mathbb Q^{-}}+x}{N}
            \right)
            \right)
            \right].
        \end{align*}
    \end{itemize}
\end{lemma}

\begin{proof}
We give the proof for \(u=u^+\). The proof for \(u=u^-\) is analogous.

(1) For \(0\leq t<1\), the solution is given by
\[
\begin{aligned}
u^+(t,x)
&=
E_N\left[
1_{[N=0]}1_{[B_1\leq0]}
+
1_{[N\neq0]}\Pi^+\left(g\left(\frac{B_1}{N}\right)\right)
\,\bigg|\, B_t=x
\right]  \\
&=
\sum_{n\geq0}p(n)
\int_{-\infty}^{\infty}
\left\{
1_{[n=0]}1_{[y\leq0]}
+
1_{[n\neq0]}\Pi^+\left(g\left(\frac{y}{n}\right)\right)
\right\}
\frac{\exp\left\{-\frac{(y-x)^2}{2\sigma^2(1-t)}\right\}}
{\sqrt{2\pi\sigma^2(1-t)}}\,dy .
\end{aligned}
\]
Equivalently, after the change of variables \(y\mapsto y+x\),
\[
u^+(t,x)
=
\sum_{n\geq0}p(n)
\int_{-\infty}^{\infty}
\left\{
1_{[n=0]}1_{[y+x\leq0]}
+
1_{[n\neq0]}\Pi^+\left(g\left(\frac{y+x}{n}\right)\right)
\right\}
\frac{\exp\left\{-\frac{y^2}{2\sigma^2(1-t)}\right\}}
{\sqrt{2\pi\sigma^2(1-t)}}\,dy .
\]
Let \(B_t=x+\sigma\beta_t^{\mathbb P}\), where \(\beta^{\mathbb P}\) is a Brownian motion under \(\mathbb P\), and suppose that \(B\) and \(N\) are independent under \(\mathbb P\). Define \(\mathbb Q^+\) on the joint space of \((B,N)\) by
\[
    \frac{d\mathbb Q^+}{d\mathbb P}
    =
    \frac{
    1_{[N=0]}1_{[B_1\leq0]}
    +
    1_{[N\neq0]}\Pi^+\left(g\left(\frac{B_1}{N}\right)\right)
    }{
    E^{\mathbb P}\left[
    1_{[N=0]}1_{[B_1\leq0]}
    +
    1_{[N\neq0]}\Pi^+\left(g\left(\frac{B_1}{N}\right)\right)
    \right]
    } .
\]
The marginal density of \(B\) is \(u^+(1,B_1)/u^+(0,B_0)\). Hence \(u^+(t,B_t)/u^+(0,B_0)\) is a bounded positive martingale under the marginal law of \(B\). Since \(u^+\) solves the heat equation, Ito's formula gives \(du^+(t,B_t)=\sigma u_x^+(t,B_t)d\beta_t^{\mathbb P}\). Therefore, by Girsanov's theorem, under the \(\mathbb Q^+\)-marginal law of \(B\),
\[
    dB_t
    =
    \sigma dW_t^{\mathbb Q^+}
    +
    \sigma^2\frac{u_x^+(t,B_t)}{u^+(t,B_t)}\,dt,
    \qquad B_0=x .
\]
This proves the SDE in \eqref{eq:SDE}. The proof for \(u=u^-\) is analogous.

(2) From Definition \ref{def:phi_g}, with \(B_1=x+\sigma\beta_1^{\mathbb P}\),
\[
\phi_g^+(x)
=
\frac{
E^{\mathbb P}\left[
E[V]1_{[N=0]}1_{[B_1\leq0]}
+
1_{[N\neq0]}\Phi^+\left(g\left(\frac{B_1}{N}\right)\right)
\right]
}{
E^{\mathbb P}\left[
1_{[N=0]}1_{[B_1\leq0]}
+
1_{[N\neq0]}\Pi^+\left(g\left(\frac{B_1}{N}\right)\right)
\right]
}.
\]
Using \(\Phi^+=\Psi^+\Pi^+\), the numerator can be written as
\[
E^{\mathbb P}\left[
\left\{
E[V]1_{[N=0]}1_{[B_1\leq0]}
+
1_{[N\neq0]}\Psi^+\left(g\left(\frac{B_1}{N}\right)\right)
\right\}
\left\{
1_{[N=0]}1_{[B_1\leq0]}
+
1_{[N\neq0]}\Pi^+\left(g\left(\frac{B_1}{N}\right)\right)
\right\}
\right].
\]
Dividing by the normalizing constant in the definition of \(\mathbb Q^+\) gives
\[
    \phi_g^+(x)
    =
    E^{\mathbb Q^+}\left[
    E[V]1_{[N=0]}1_{[B_1\leq0]}
    +
    1_{[N\neq0]}\Psi^+\left(g\left(\frac{B_1}{N}\right)\right)
    \right].
\]
The proof for \(\phi_g^-\) is analogous.

(3) At \(x=0\), Definition \ref{def:phi_g} gives
\[
\phi_g^+(0)
=
\frac{
\frac{p(0)}{2}E[V]
+
\sum_{n\geq1}p(n)\int_{-\infty}^{\infty}
\Phi^+\left(g\left(\frac{z}{n}\right)\right)q(\sigma,z)\,dz
}{
\frac{p(0)}{2}
+
\sum_{n\geq1}p(n)\int_{-\infty}^{\infty}
\Pi^+\left(g\left(\frac{z}{n}\right)\right)q(\sigma,z)\,dz
}.
\]
The expression for \(\phi_g^-(0)\) is analogous, with \(\Phi^+\) and \(\Pi^+\) replaced by \(\Phi^-\) and \(\Pi^-\). Using \(\Phi^++\Phi^-=E[V]\) and \(\Pi^++\Pi^-=1\), the inequality \(\phi_g^+(0)>\phi_g^-(0)\) is equivalent to
\[
    \sum_{n\geq1}p(n)
    \int_{-\infty}^{\infty}
    \left\{
    \Psi^+\left(g\left(\frac{z}{n}\right)\right)-E[V]
    \right\}
    \Pi^+\left(g\left(\frac{z}{n}\right)\right)
    q(\sigma,z)\,dz
    >0 .
\]
This holds because \(\Psi^+(y)\geq E[V]\) for \(y\in[m,M)\), with strict inequality on a set of positive Gaussian measure whenever \(p(0)<1\) and \(V\) is nondegenerate.

(4) Suppose now that \(g\) is nondecreasing. From the representation of \(u^+\) above, the terminal function is nonincreasing in \(x\), and hence \(u_x^+(t,x)\leq0\). The drift \(\sigma^2u_x^+(t,x)/u^+(t,x)\) in the SDE is therefore nonpositive. By the standard comparison theorem for one dimensional SDEs, a larger initial condition leads to a stochastically larger solution. Since, for each \(n\geq1\), the map \(b\mapsto\Psi^+(g(b/n))\) is nondecreasing, the representation in part (2) implies that \(\phi_g^+\) is nondecreasing. The same argument applied to \(u^-\) shows that \(\phi_g^-\) is nondecreasing, and therefore \(\phi_g\) is nondecreasing.

Finally, since \(u_x^+\leq0\), the solution of the SDE for \(u=u^+\) is dominated by the Brownian motion with the same initial condition, namely \(x+\sigma W_1^{\mathbb Q^+}\). Using again the monotonicity of the terminal payoff gives
\[
\phi_g^+(x)
\leq
E^{\mathbb Q^+}\left[
E[V]1_{[N=0]}1_{[\sigma W_1^{\mathbb Q^+}+x\leq0]}
+
1_{[N\neq0]}\Psi^+\left(
g\left(\frac{\sigma W_1^{\mathbb Q^+}+x}{N}\right)
\right)
\right].
\]
The lower bound for \(\phi_g^-\) follows analogously.
This proves all four claims.
\end{proof}

We now apply the preceding monotonicity property to non-decreasing solutions of the fixed point equation. The next lemma shows that such a solution has the correct boundary behaviour and is strictly increasing.

\begin{lemma}\label{lemma:F_limits}
Let $F$ be a non-decreasing solution of \eqref{eq:F_FixedPointMapping}. Assume that $F$ takes values in $[m,M]$, where $m<M<\infty$. Then $F$ is strictly increasing, and
$\lim_{x\rightarrow\infty}F(x)=M$ and $\lim_{x\rightarrow-\infty}F(x)=m$.
\end{lemma}

\begin{proof}
Since $F$ is non-decreasing and takes values in $[m,M]$, the limits
$L:=\lim_{x\rightarrow\infty}F(x)$ and
$\ell:=\lim_{x\rightarrow-\infty}F(x)$ exist, with $m\leq \ell\leq L\leq M$. Moreover, since $\phi_F^+(x)$ and $\phi_F^-(x)$ are conditional expectations of $V$, and $V$ takes values in $[m,M]$, we have $m\leq \phi_F(x)\leq M$ for all $x\neq0$. Hence $\phi_F$ is bounded.

We first prove that $L=M$. We use the following Gaussian shift argument: if $G$ is bounded and $G(u)\rightarrow G_\infty$ as $u\rightarrow\infty$, then, for each fixed $n\geq1$,
\[
    n\int_{-\infty}^{\infty}q(\sigma,s-nu)G(u)\,du
    =
    \int_{-\infty}^{\infty}q(\sigma,w)G\left(\frac{s-w}{n}\right)\,dw
    \longrightarrow G_\infty
\]
as $s\rightarrow\infty$, by dominated convergence.

Suppose, for contradiction, that $L<M$. From Definition \ref{def:phi_g},
\[
\phi_F^+(s)
=
\frac{
p(0)E[V]P(Z\geq s)
+
\sum_{n\geq1}p(n)n\int_{-\infty}^{\infty}q(\sigma,s-nu)\Phi^+(F(u))\,du
}{
p(0)P(Z\geq s)
+
\sum_{n\geq1}p(n)n\int_{-\infty}^{\infty}q(\sigma,s-nu)\Pi^+(F(u))\,du
}.
\]
Since $F(u)\rightarrow L$, we have $\Phi^+(F(u))\rightarrow\Phi^+(L-)$ and $\Pi^+(F(u))\rightarrow\Pi^+(L-)$. Moreover, $\Pi^+(L-)=P(V\geq L)>0$, $p(0)E[V]P(Z\geq s)\rightarrow0$ and $p(0)P(Z\geq s)\rightarrow0$ as $s\rightarrow\infty$. Hence, using the Gaussian shift argument and dominated convergence over $n$,
\[
    \lim_{s\rightarrow\infty}\phi_F^+(s)
    =
    \frac{\Phi^+(L-)}{\Pi^+(L-)}
    =
    \Psi^+(L-).
\]

Now take $x\rightarrow\infty$ in \eqref{eq:F_FixedPointMapping}. For the first kernel term,
\[
    \int_{-\infty}^{\infty}q(\sigma,z-nx)\phi_F(z)\,dz
    =
    \int_{-\infty}^{\infty}q(\sigma,w)\phi_F(w+nx)\,dw
    \longrightarrow \Psi^+(L-).
\]
For the second kernel term,
\[
\begin{aligned}
    \int_{-\infty}^{\infty}\bar q(\sigma,n,x,z)\phi_F(z)\,dz
    =
    \frac{1}{x}\int_0^x
    \int_{-\infty}^{\infty}q(\sigma,w)\phi_F(w+ny)\,dw\,dy 
    \longrightarrow \Psi^+(L-),
\end{aligned}
\]
because the inner integral converges to $\Psi^+(L-)$ as $y\rightarrow\infty$. Therefore, for each fixed $n\geq1$,
\[
\begin{aligned}
&\int_{-\infty}^{\infty}
\left\{
\frac{1}{n}q(\sigma,z-nx)
+
\frac{n-1}{n}\bar q(\sigma,n,x,z)
\right\}\phi_F(z)\,dz \\
&\qquad\longrightarrow
\frac{1}{n}\Psi^+(L-)+\frac{n-1}{n}\Psi^+(L-)
=
\Psi^+(L-).
\end{aligned}
\]
Since $\phi_F$ is bounded and $\sum_{n\geq1}p(n)/(1-p(0))=1$, dominated convergence over $n$ gives $L=\Psi^+(L-)$. But for every $y<M$, $\Psi^+(y-)=E[V\mid V\geq y]>y$, because $P(V>y)>0$. This contradicts $L<M$. Hence $L=M$.
Similarly for the lower limit that $\ell=m$. 


Finally, since $F$ is non-decreasing, Lemma \ref{lemma:phi_property}(4) implies that $\phi_F$ is non-decreasing. The preceding argument also gives the corresponding tail limits of $\phi_F$, so $\phi_F$ is not constant. Hence, by the convolution representation in \eqref{eq:F_FixedPointMapping} and the strict positivity of the Gaussian density, $F$ is strictly increasing.
\end{proof}

Lemma~\ref{lemma:F_limits} shows that the marginal trading cost $F$ increases strictly with trade size $x$. As the buy order becomes arbitrarily large, the marginal trading cost converges to $M$, the highest possible fundamental value of the asset. As the sell order becomes arbitrarily large, it converges to $m$, the lowest possible fundamental value.

\begin{theorem}\label{thm:equilibrium_existence}
Suppose \(-\infty<m<M<\infty\), \(p(0)<1\), and \(E[N]<\infty\). Then there exists an equilibrium.
\end{theorem}

\begin{proof}
We prove the result by applying Schauder's fixed point theorem. For \(x\neq0\), observe that \(\bar q(\sigma,n,x,z)=x^{-1}\int_0^x q(\sigma,z-ny)\,dy\). Hence
\[
\begin{aligned}
    \frac{\partial}{\partial x}\bar q(\sigma,n,x,z)
    &=
    \frac{q(\sigma,z-nx)-\bar q(\sigma,n,x,z)}{x}  \\
    &=
    -\frac{n}{x^2}\int_0^x y q_y(\sigma,z-ny)\,dy .
\end{aligned}
\]
Moreover, \(|\partial q(\sigma,z-nx)/\partial x|=n|q_y(\sigma,z-nx)|\). Since \(V\) takes values in \([m,M]\), the function \(\phi_g\) also takes values in \([m,M]\), and hence \(|\phi_g|\leq |m|+M\). Using \(\int_{-\infty}^{\infty}|q_y(\sigma,z-a)|\,dz=\sigma^{-1}\sqrt{2/\pi}\), uniformly in \(a\), the preceding identities give
\[
\begin{aligned}
\left|\frac{d}{dx}Tg(x)\right|
&\leq
(|m|+M)\frac{1}{\sigma}\sqrt{\frac{2}{\pi}}
\sum_{n\geq1}\frac{p(n)}{1-p(0)}
\left(1+\frac{n-1}{2}\right)  \\
&=:K_0<\infty .
\end{aligned}
\]
The finiteness follows from \(E[N]<\infty\).

Let \(\mathcal X:=L^2(\mathbb R,\mu_0)\), where \(\mu_0(dx)=(2\pi)^{-1/2}e^{-x^2/2}\,dx\). Define
\[
    D_0
    :=
    \left\{
    g:\mathbb R\mapsto[m,M]:
    g \text{ is nondecreasing and }
    |g(x)-g(y)|\leq K_0|x-y|,\ \text{for all }x,y\in\mathbb R
    \right\},
\]
and
\[
    D
    :=
    \left\{
    g\in\mathcal X:
    g=g_0,\ \mu_0\text{-a.e. for some }g_0\in D_0
    \right\}.
\]
Then \(D\) is nonempty and convex.

For \(g\in D\), let \(g_0\in D_0\) be its continuous representative and define
\[
    Tg(x)
    :=
    \sum_{n\geq1}\frac{p(n)}{1-p(0)}
    \int_{-\infty}^{\infty}
    \left\{
    \frac{1}{n}q(\sigma,z-nx)
    +
    \frac{n-1}{n}\bar q(\sigma,n,x,z)
    \right\}
    \phi_{g_0}(z)\,dz .
\]
For each fixed \(n\) and \(x\), the measure
\[
    \left\{
    \frac{1}{n}q(\sigma,z-nx)
    +
    \frac{n-1}{n}\bar q(\sigma,n,x,z)
    \right\}\,dz
\]
has total mass one. Since \(\phi_{g_0}\) takes values in \([m,M]\), \(Tg\) also takes values in \([m,M]\). The derivative estimate above gives \(|Tg(x)-Tg(y)|\leq K_0|x-y|\) for all \(x,y\in\mathbb R\).

Moreover, since \(g_0\) is nondecreasing, Lemma~\ref{lemma:phi_property} implies that \(\phi_{g_0}\) is nondecreasing. For each \(n\geq1\), define \(A_n(x):=\int_{-\infty}^{\infty}q(\sigma,z-nx)\phi_{g_0}(z)\,dz\). Then \(A_n\) is nondecreasing in \(x\). Also, for \(x\neq0\),
\[
    \int_{-\infty}^{\infty}\bar q(\sigma,n,x,z)\phi_{g_0}(z)\,dz
    =
    \frac{1}{x}\int_0^x A_n(y)\,dy,
\]
with the value at \(x=0\) defined by continuity. Since \(A_n\) is nondecreasing, the map \(x\mapsto x^{-1}\int_0^x A_n(y)\,dy\) is also nondecreasing. Hence \(Tg\) is nondecreasing. Therefore \(Tg\in D_0\), and consequently \(T(D)\subset D\).

We next show that \(D\) is compact in \(\mathcal X\). Let \((g_j)\subset D\), and choose representatives \(g_j^0\in D_0\). The family \((g_j^0)\) is uniformly bounded and equicontinuous. By the Arzela Ascoli theorem on compact intervals and a diagonal argument, there exists a subsequence, still denoted by \((g_j^0)\), which converges uniformly on compact subsets of \(\mathbb R\) to some function \(g^0\). The limit satisfies \(g^0\in D_0\). Since the sequence is uniformly bounded and \(\mu_0\) is a probability measure, dominated convergence gives \(g_j^0\rightarrow g^0\) in \(L^2(\mathbb R,\mu_0)\). Hence \(D\) is compact.

It remains to show that \(T:D\mapsto D\) is continuous. Suppose \(g_j\rightarrow g\) in \(D\). By compactness of \(D\), every subsequence of \((g_j)\) has a further subsequence whose continuous representatives converge locally uniformly to a representative of \(g\). Along such a subsequence, \(\phi_{g_j^0}(z)\rightarrow\phi_{g_0}(z)\) for almost every \(z\). Since \(\phi_{g_j^0}\) is uniformly bounded by \(|m|+M\), dominated convergence gives \(Tg_j(x)\rightarrow Tg(x)\) for every \(x\in\mathbb R\). Another application of dominated convergence with respect to \(\mu_0\) gives \(Tg_j\rightarrow Tg\) in \(L^2(\mathbb R,\mu_0)\). Therefore \(T\) is continuous.

Thus \(D\) is a nonempty compact convex subset of \(\mathcal X\), and \(T:D\mapsto D\) is continuous. By Schauder's fixed point theorem, there exists \(F\in D\) such that \(F=TF\). Hence \eqref{eq:F_FixedPointMapping} admits a nondecreasing solution.

By Lemma~\ref{lemma:F_limits}, this solution is strictly increasing and satisfies \(\lim_{x\rightarrow-\infty}F(x)=m\) and \(\lim_{x\rightarrow\infty}F(x)=M\). In particular, \(F(x)\in(m,M)\) for every finite \(x\). Therefore \(F^{-1}(V)\) is well defined, and by the equilibrium characterization above, the corresponding pair \((F^{-1}(V),\phi_F)\) constitutes an equilibrium.
\end{proof}

\begin{remark}
Theorem~\ref{thm:equilibrium_existence} establishes equilibrium existence under the bounded-support assumption on $V$. This assumption is imposed for the fixed point argument above. In Section~\ref{section:Numerics}, we also consider unbounded asset value distributions. The numerical evidence suggests that the fixed point equation continues to admit a unique solution, and hence an equilibrium, beyond the bounded-support case.
\end{remark}

\section{Price Impact Asymptotics}\label{section:PriceAsymptotics}

The shape of price impact as a function of trading volume plays an important role in trading strategies and market liquidity. In this section, we study the asymptotic shape of price impact implied by the equilibrium model.

\begin{definition}
    A function $g:(0,\infty)\rightarrow(0,\infty)$ is said to be \textit{regularly varying} of index $\rho$ at $\infty$ if
    \[
        \lim_{\lambda\rightarrow\infty}\frac{g(\lambda x)}{g(\lambda)}=x^{\rho},
        \qquad x>0.
    \]
    Similarly, a function $g:(-\infty,0)\rightarrow(0,\infty)$ is said to be regularly varying of index $\rho$ at $-\infty$ if $g(-x)$ is regularly varying of index $\rho$ at $\infty$. When $\rho=0$, that is, when $\lim_{\lambda\rightarrow\infty}g(\lambda x)/g(\lambda)=1$ for all $x>0$, $g$ is said to be \textit{slowly varying}.
\end{definition}

\begin{remark}
    Roughly speaking, a regularly varying function behaves asymptotically like a power function. Slowly varying functions include, for example, $\log x$, $(\log x)^{\alpha}$ for $\alpha\in\mathbb{R}$, $\log\log x$, and $\exp\{(\log x)^{\alpha}\}$ for $\alpha\in(0,1)$.
\end{remark}

\begin{assumption}\label{assumption:Psi}
    $\Pi^{+}$ and $\Pi^{-}$ have continuous derivatives on $(m,M)$. Moreover, the limits
    \[
        \Psi_{x}^{+}(M):=\lim_{x\rightarrow M}\frac{d}{dx}\Psi^{+}(x),
        \qquad
        \Psi_{x}^{-}(m):=\lim_{x\rightarrow m}\frac{d}{dx}\Psi^{-}(x)
    \]
    exist and coincide with the left and right endpoint derivatives, respectively. In addition, $\Psi_{x}^{+}(M)\Psi_{x}^{-}(m)\neq0$.
\end{assumption}

Note that $\Psi_x^+(M)$ and $\Psi_x^-(m)$ describe the tail shape of the trading asset distribution near the endpoints. Specifically, taking the right tail as an example, suppose that $\Pi^+(x)=P(V>x)\propto (M-x)^\alpha$ as $x\rightarrow M$, where $\alpha>0$. Then $\Pi_x^+(x)/\Pi^+(x)\sim-\alpha(M-x)^{-1}$. Moreover, by integration by parts, $\Psi^+(x)-x=\int_x^M \Pi^+(y)\,dy/\Pi^+(x)$. Hence
\[
    \frac{M-\Psi^+(x)}{M-x}
    =
    1-
    \frac{\int_x^M \Pi^+(y)\,dy}{(M-x)\Pi^+(x)}.
\]
Applying L'Hopital's rule as $x\rightarrow M$ gives $\Psi_x^+(M)=\alpha/(1+\alpha)<1$. This shows that $\Psi_x^+(M)<1$ captures a power type upper endpoint tail.

In the following theorem, we show that if the distribution of the trading asset has power type endpoint tails, then the equilibrium price impact follows a power law, with the exponent characterised by a fixed point equation.

\begin{theorem}\label{thm:powerlaw_PriceImpact}
Assume Assumption~\ref{assumption:Psi} holds. Suppose $N$ has probability mass function $p(n)$, $n=0,1,2,\dots$, with $p(1)>0$, and suppose $V$ is supported on a bounded interval $(m,M)\subset\mathbb{R}$. Let $F$ be a strictly increasing solution to \eqref{eq:F_FixedPointMapping}. Suppose $0<\Psi_x^+(M)<1$ and $0<\Psi_x^-(m)<1$. Define
\[
    k_+:=\frac{\Psi_x^+(M)}{1-\Psi_x^+(M)},
    \qquad
    k_-:=\frac{\Psi_x^-(m)}{1-\Psi_x^-(m)},
\]
and assume $E\left[N^{1+(k_+\vee k_-)}\right]<\infty$. Then the following hold:
    
\begin{enumerate}[label=(\arabic*)]
    \item $M-F$ is regularly varying at $\infty$ with index $\rho^+\in(-1,0)$, where $\rho^+$ satisfies $\mathcal{T}^+(\rho^+)=\rho^+$, with
    \begin{align}\label{eq:PowerLaw_fixedpoint_plus}
        \mathcal{T}^+(x)
        :=
        \frac{\Psi_x^+(M)}{1-p(0)}
        \frac{
        E\left[1_{[N\neq0]}N^{-(k_++1)x}\right]
        }{
        E\left[1_{[N\neq0]}N^{-k_+x}\right]
        }
        E\left[
        1_{[N\neq0]}\left(xN^{x-1}+N^x\right)
        \right]
        -1.
    \end{align}
    Moreover, $\Pi^+(F)$ is regularly varying at $\infty$ with index $k_+\rho^+$.

    \item $F-m$ is regularly varying at $-\infty$ with index $\rho^-\in(-1,0)$, where $\rho^-$ satisfies $\mathcal{T}^-(\rho^-)=\rho^-$, with
    \begin{align}\label{eq:PowerLaw_fixedpoint_minus}
        \mathcal{T}^-(x)
        :=
        \frac{\Psi_x^-(m)}{1-p(0)}
        \frac{
        E\left[1_{[N\neq0]}N^{-(k_-+1)x}\right]
        }{
        E\left[1_{[N\neq0]}N^{-k_-x}\right]
        }
        E\left[
        1_{[N\neq0]}\left(xN^{x-1}+N^x\right)
        \right]
        -1.
    \end{align}
    Moreover, $\Pi^-(F)$ is regularly varying at $-\infty$ with index $k_-\rho^-$.
\end{enumerate}
\end{theorem}

\begin{remark}
The moment condition $E\left[N^{1+(k_+\vee k_-)}\right]<\infty$ is used only to justify the interchange of limits and summation over $n$ in the asymptotic argument. It is automatically satisfied when $N$ has finite support, and also holds for standard count distributions with exponentially decaying tails, such as the Poisson, geometric, and negative binomial distributions. The condition mainly rules out very heavy-tailed specifications for the number of insiders.
\end{remark}

\begin{proof}
We prove the first statement. The proof of the second statement is analogous. Let $G(x):=M-F(x)$ for $x>0$. By Lemma~\ref{lemma:F_limits}, $G(x)>0$ and $G(x)\rightarrow0$ as $x\rightarrow\infty$. For $\alpha>0$, define
\[
    \gamma_\alpha(x):=\frac{G(\alpha x)}{G(\alpha)},\qquad x>0.
\]
From the fixed point equation~\eqref{eq:F_FixedPointMapping}, after the change of variables $z\rightarrow\alpha z$,
\[
\begin{aligned}
\gamma_\alpha(x)
&=
\sum_{n\geq1}\frac{p(n)}{1-p(0)}
\int_{-\infty}^{\infty}
\left\{
\frac{1}{n}q\left(\frac{\sigma}{\alpha},z-nx\right)
+
\frac{n-1}{n}\bar q\left(\frac{\sigma}{\alpha},n,x,z\right)
\right\}
\frac{M-\phi_F(\alpha z)}{G(\alpha)}\,dz .
\end{aligned}
\]
The moment condition $E[N^{1+(k_+\vee k_-)}]<\infty$ provides the domination needed to interchange limits and summation over $n$.

We first identify the limiting form of $M-\phi_F^+$. For $y>0$, the mean value theorem gives
\[
    M-\Psi^+(F(\alpha y))
    =
    \Psi_x^+(\xi_{\alpha,y})G(\alpha y),
    \qquad
    \xi_{\alpha,y}\in(F(\alpha y),M).
\]
Thus, along any subsequence such that $\gamma_\alpha(y)\rightarrow\gamma(y)$,
\begin{align}\label{eq:Psi_scaling_limit}
    \frac{M-\Psi^+(F(\alpha y))}{G(\alpha)}
    \rightarrow
    \Psi_x^+(M)\gamma(y).
\end{align}
Moreover,
\[
    \Psi^+(x)-x
    =
    \frac{\int_x^M \Pi^+(u)\,du}{\Pi^+(x)},
    \qquad
    -\frac{\Pi_x^+(x)}{\Pi^+(x)}
    =
    \frac{\Psi_x^+(x)}{\Psi^+(x)-x}.
\]
Since $\Psi_x^+(M)\in(0,1)$, the local endpoint tail exponent is
\[
    k_+ = \frac{\Psi_x^+(M)}{1-\Psi_x^+(M)}.
\]
Consequently, along the same subsequence,
\[
    \frac{\Pi^+(F(\alpha y))}{\Pi^+(F(\alpha))}
    \rightarrow
    \gamma(y)^{k_+},
    \qquad y>0.
\]

Using Definition~\ref{def:phi_g}, changing variables $y=\alpha r$, and using the concentration of $q(\sigma/(\alpha n),r-z/n)\,dr$ at $z/n$, we obtain, for $z>0$,
\begin{align}\label{eq:H_gamma_limit}
\frac{M-\phi_F^+(\alpha z)}{G(\alpha)}
\rightarrow
H_\gamma(z)
:=
\Psi_x^+(M)
\frac{
\sum_{m\geq1}p(m)\gamma(z/m)^{k_++1}
}{
\sum_{m\geq1}p(m)\gamma(z/m)^{k_+}
}.
\end{align}
Substituting \eqref{eq:H_gamma_limit} into the scaled fixed point equation shows that every subsequential scaling limit $\gamma$ satisfies
\begin{align}\label{eq:gamma_limiting_equation}
\gamma(x)
&=
\sum_{n\geq1}\frac{p(n)}{1-p(0)}
\left\{
\frac{1}{n}H_\gamma(nx)
+
\frac{n-1}{nx}\int_0^x H_\gamma(ny)\,dy
\right\},
\qquad x>0.
\end{align}

The comparison method used in Theorem~B.1 of \cite{Cetin_Waelbroeck_2023} applies to the limiting equation \eqref{eq:gamma_limiting_equation}. In the present setting, the deterministic number of insiders is replaced by averaging with respect to the distribution of $N$, and the moment condition above justifies the corresponding interchange of limits and summation. Hence every positive locally bounded subsequential scaling limit is a power function. Therefore $\gamma(x)=x^{\rho^+}$ for some $\rho^+\in[-1,0)$, and since every subsequential scaling limit has this form, $G=M-F$ is regularly varying at $\infty$ with index $\rho^+$.

It remains to identify $\rho^+$. Substituting $\gamma(x)=x^{\rho^+}$ into \eqref{eq:H_gamma_limit} gives
\begin{align}\label{eq:M_phi_frac}
    H_\gamma(z)
    =
    \Psi_x^+(M)z^{\rho^+}
    \frac{
    E\left[1_{[N\neq0]}N^{-(k_++1)\rho^+}\right]
    }{
    E\left[1_{[N\neq0]}N^{-k_+\rho^+}\right]
    } .
\end{align}
Combining \eqref{eq:M_phi_frac} with \eqref{eq:gamma_limiting_equation}, for $x>0$,
\[
\begin{aligned}
x^{\rho^+}
&=
\Psi_x^+(M)
\frac{
E\left[1_{[N\neq0]}N^{-(k_++1)\rho^+}\right]
}{
E\left[1_{[N\neq0]}N^{-k_+\rho^+}\right]
}
\sum_{n\geq1}\frac{p(n)}{1-p(0)}
\left\{
\frac{1}{n}(nx)^{\rho^+}
+
\frac{n-1}{nx}\int_0^x(ny)^{\rho^+}\,dy
\right\}.
\end{aligned}
\]
Since $\rho^+>-1$,
\[
    \frac{n-1}{nx}\int_0^x(ny)^{\rho^+}\,dy
    =
    \frac{n-1}{n}\frac{n^{\rho^+}x^{\rho^+}}{\rho^++1}.
\]
Dividing by $x^{\rho^+}$ and rearranging gives
\[
\rho^+
=
\frac{\Psi_x^+(M)}{1-p(0)}
\frac{
E\left[1_{[N\neq0]}N^{-(k_++1)\rho^+}\right]
}{
E\left[1_{[N\neq0]}N^{-k_+\rho^+}\right]
}
E\left[
1_{[N\neq0]}\left(\rho^+N^{\rho^+-1}+N^{\rho^+}\right)
\right]
-1.
\]
Thus $\rho^+$ satisfies $\mathcal{T}^+(\rho^+)=\rho^+$. Finally,
\[
    \frac{\Pi^+(F(\alpha y))}{\Pi^+(F(\alpha))}
    \rightarrow
    y^{k_+\rho^+}
\]
shows that $\Pi^+(F)$ is regularly varying at $\infty$ with index $k_+\rho^+$.
\end{proof}

Theorem~\ref{thm:powerlaw_PriceImpact} characterises the power law exponents through the fixed point equations \eqref{eq:PowerLaw_fixedpoint_plus} and \eqref{eq:PowerLaw_fixedpoint_minus}. These equations show that the asymptotic shape of the equilibrium limit order book depends jointly on the endpoint behaviour of the asset value and on the distribution of informed trading participation. In particular, the effect of informed-trader-count uncertainty is not summarized by the expected number of informed traders alone. Markets with the same average level of informed participation may have different large order price impact exponents when probability mass is distributed differently across possible insider counts. 

In the deterministic case, the fixed point equations reduce to affine equations, and the fixed points are unique. For general distributions of \(N\), although uniqueness is not established theoretically in full generality, it is observed in all numerical specifications considered below. We therefore impose uniqueness in the following comparative statics result. The proposition shows that a larger endpoint tail index, corresponding to a thinner asset tail near the endpoint, leads to a power law exponent with smaller magnitude. That is, large order flow is less indicative of extreme asset value realisations, so liquidity suppliers face weaker adverse selection and adjust prices less aggressively. In this sense, a larger endpoint tail index is associated with a flatter limit order book far from the origin.

\begin{proposition}\label{prop:tail_index_price_impact_exponent}
Assume the conditions of Theorem~\ref{thm:powerlaw_PriceImpact} hold. Fix the distribution of \(N\), and assume that the fixed point equations \eqref{eq:PowerLaw_fixedpoint_plus} and \eqref{eq:PowerLaw_fixedpoint_minus} admit unique solutions in \((-1,0)\) for the endpoint tail indices under consideration. Then \(\rho^+\) is strictly increasing in \(\Psi_x^+(M)\), and \(\rho^-\) is strictly increasing in \(\Psi_x^-(m)\).
\end{proposition}

\begin{proof}
We prove the upper tail case; the lower tail case is identical. Write
\[
    \psi:=\Psi_x^+(M),
    \qquad
    k:=\frac{\psi}{1-\psi},
\]
and let \(\mathcal T_\psi\) denote the map in \eqref{eq:PowerLaw_fixedpoint_plus}. Set \(g_\psi(\rho):=\mathcal T_\psi(\rho)-\rho\), \(\rho\in(-1,0)\).

Fix \(\rho\in(-1,0)\). In \eqref{eq:PowerLaw_fixedpoint_plus}, the factor \(\psi/(1-p(0))\) is strictly increasing in \(\psi\), while
\[
    \rho N^{\rho-1}+N^\rho=N^{\rho-1}(N+\rho)>0
    \qquad \text{on } \{N\neq0\}.
\]
It remains to consider
\[
    R_k(\rho):=
    \frac{E\left[1_{\{N\neq0\}}N^{-(k+1)\rho}\right]}
    {E\left[1_{\{N\neq0\}}N^{-k\rho}\right]}.
\]
We show that \(R_k(\rho)\) is nondecreasing in \(k\). Differentiating \(\log R_k(\rho)\), and defining the probability measure \(Q_k\) on \(\sigma(N)\) by
\[
    Q_k(A):=
    \frac{E\left[1_A1_{\{N\neq0\}}N^{-k\rho}\right]}
    {E\left[1_{\{N\neq0\}}N^{-k\rho}\right]},
    \qquad A\in\sigma(N),
\]
gives
\[
\begin{aligned}
\frac{d}{dk}\log R_k(\rho)
&=
(-\rho)
\left\{
\frac{E_{Q_k}\left[N^{-\rho}\log N\right]}{E_{Q_k}\left[N^{-\rho}\right]}
-
E_{Q_k}\left[\log N\right]
\right\} \\
&=
(-\rho)
\frac{
\operatorname{Cov}_{Q_k}\left(N^{-\rho},\log N\right)
}{
E_{Q_k}\left[N^{-\rho}\right]
}.
\end{aligned}
\]
The denominator is positive. Since \(-\rho>0\), both \(n\mapsto n^{-\rho}\) and \(n\mapsto\log n\) are increasing on \(\{1,2,\ldots\}\), and hence their covariance under \(Q_k\) is nonnegative. Thus \(R_k(\rho)\) is nondecreasing in \(k\). Since \(k=\psi/(1-\psi)\) is strictly increasing in \(\psi\), it follows that \(g_\psi(\rho)\) is strictly increasing in \(\psi\) for every fixed \(\rho\in(-1,0)\).

Let \(0<\psi_1<\psi_2<1\), and let \(\rho_1,\rho_2\in(-1,0)\) be the unique zeros of \(g_{\psi_1}\) and \(g_{\psi_2}\), respectively. The pointwise monotonicity gives
\[
    g_{\psi_2}(\rho_1)>g_{\psi_1}(\rho_1)=0.
\]
Moreover,
\[
    \lim_{\rho\uparrow0}g_{\psi_2}(\rho)=\psi_2-1<0.
\]
By continuity, \(g_{\psi_2}\) has a zero in \((\rho_1,0)\). By uniqueness, this zero is \(\rho_2\), and hence \(\rho_2>\rho_1\). Therefore \(\rho^+\) is strictly increasing in \(\Psi_x^+(M)\). The assertion for \(\rho^-\) follows in the same way.
\end{proof}

Theorem~\ref{thm:powerlaw_PriceImpact} also implies that the marginal price function $h$ and the marginal cost function $F$ have the same regular variation order in the right tail. Indeed, since $h(x)=\phi_F^+(x)$ for $x>0$, the argument leading to \eqref{eq:M_phi_frac} gives
\begin{equation}\label{eq:h_F_ratio}
    \lim_{x\rightarrow\infty}
    \frac{M-h(x)}{M-F(x)}
    =
    \Psi_x^+(M)
    \frac{
    E\left[1_{[N\neq0]}N^{-(k_++1)\rho^+}\right]
    }{
    E\left[1_{[N\neq0]}N^{-k_+\rho^+}\right]
    } .
\end{equation}
Thus $M-h$ is regularly varying at $\infty$ with the same index $\rho^+$ as $M-F$. The uncertainty in the number of informed traders affects the asymptotic level through the correction factor
\[
    \frac{
    E\left[1_{[N\neq0]}N^{-(k_++1)\rho^+}\right]
    }{
    E\left[1_{[N\neq0]}N^{-k_+\rho^+}\right]
    },
\]
which depends on the distribution of $N$. In the deterministic case $N=n>1$, this factor reduces to $n^{-\rho^+}$, and the fixed point equation in \eqref{eq:PowerLaw_fixedpoint_plus} collapses to the deterministic exponent formula derived in \cite{Cetin_Waelbroeck_2023}. Hence the random-insider specification preserves the power-law asymptotic structure while showing explicitly how uncertainty about informed trading changes the tail constant and the admissible price impact exponent.

We are also interested in the asymptotic behaviour of the implementation shortfall. When the number of informed traders is random, the implementation shortfall describes the difference between a paper trading benchmark and the actual trading cost.

\begin{definition}
The implementation shortfall associated with an individual insider trading $x$ units is defined by
\[
    IS(x)
    :=
    E\left[
    \frac{1}{Nx}\int_0^{Nx}h(Z+u)\,du
    \,\Bigg|\, N\geq1
    \right].
\]
\end{definition}

Since informed traders are symmetric in equilibrium, each informed trader submits the same order size. The expectation is taken with respect to the number of informed traders $N$ and the noise order $Z$, which are independent. By the change of variables $u=Nv$, we equivalently have
\[
    IS(x)
    =
    \frac{1}{x}\int_0^x
    E\left[h(Z+Nv)\mid N\geq1\right]\,dv .
\]
Using this representation, we obtain the following asymptotic relationship between the implementation shortfall in equilibrium $IS^*$ and the marginal cost function $F$. 
The result shows that the average execution cost faced by an informed trader inherits the same power law decay as the marginal cost, with an adjustment reflecting how the trader's order is aggregated with other informed orders.

\begin{corollary}\label{corollary:IS(x)}
Assume the conditions of Theorem~\ref{thm:powerlaw_PriceImpact}. Let $(h^*,x^*)$ be an equilibrium, and let $F$ be the corresponding solution of \eqref{eq:F_FixedPointMapping}. Then
\[
\begin{aligned}
    M-IS^*(x)
    &\sim
    \frac{\Psi_x^+(M)}{1-p(0)}
    \frac{
    E\left[1_{[N\neq0]}N^{-(k_++1)\rho^+}\right]
    }{
    E\left[1_{[N\neq0]}N^{-k_+\rho^+}\right]
    }
    \frac{
    E\left[1_{[N\neq0]}N^{\rho^+}\right]
    }{\rho^++1}
    \{M-F(x)\},
    \qquad x\rightarrow\infty, \\
    IS^*(x)-m
    &\sim
    \frac{\Psi_x^-(m)}{1-p(0)}
    \frac{
    E\left[1_{[N\neq0]}N^{-(k_-+1)\rho^-}\right]
    }{
    E\left[1_{[N\neq0]}N^{-k_-\rho^-}\right]
    }
    \frac{
    E\left[1_{[N\neq0]}N^{\rho^-}\right]
    }{\rho^-+1}
    \{F(x)-m\},
    \qquad x\rightarrow-\infty .
\end{aligned}
\]
\end{corollary}

\begin{proof}
We shall provide the proof of the first statement. The second statement is analogous. From the representation
\[
    IS^*(x)
    =
    E\left[
    \frac{1}{x}\int_0^x h^*(Z+Nu)\,du
    \,\Bigg|\, N\geq1
    \right],
\]
and using $h^*(y)=\phi_{F}^+(y)$ for $y>0$, we obtain, for $x>0$,
\[
\begin{aligned}
\frac{M-IS^*(x)}{M-F(x)}
&=
\frac{1}{x}
\sum_{n\geq1}\frac{p(n)}{1-p(0)}
\int_0^x
\int_{-\infty}^{\infty}
q(\sigma,z-nu)
\frac{M-\phi_{F}(z)}{M-F(x)}
\,dz\,du \\
&=
\sum_{n\geq1}\frac{p(n)}{1-p(0)}
\int_0^1
\int_{-\infty}^{\infty}
q(\sigma,z-nux)
\frac{M-\phi_{F}(z)}{M-F(x)}
\,dz\,du \\
&=
\sum_{n\geq1}\frac{p(n)}{1-p(0)}
\int_0^1
\int_{-\infty}^{\infty}
q\left(\frac{\sigma}{x},z-nu\right)
\frac{M-\phi_{F}(xz)}{M-F(x)}
\,dz\,du .
\end{aligned}
\]
As $x\rightarrow\infty$, the measure $q(\sigma/x,z-nu)\,dz$ converges weakly to the point mass at $nu$. Combining this with \eqref{eq:M_phi_frac}, we have
\[
\begin{aligned}
\lim_{x\rightarrow\infty}
\frac{M-IS^*(x)}{M-F(x)}
&=
\Psi_x^+(M)
\frac{
E\left[1_{[N\neq0]}N^{-(k_++1)\rho^+}\right]
}{
E\left[1_{[N\neq0]}N^{-k_+\rho^+}\right]
}
\sum_{n\geq1}\frac{p(n)}{1-p(0)}
\int_0^1(nu)^{\rho^+}\,du \\
&=
\frac{\Psi_x^+(M)}{1-p(0)}
\frac{
E\left[1_{[N\neq0]}N^{-(k_++1)\rho^+}\right]
}{
E\left[1_{[N\neq0]}N^{-k_+\rho^+}\right]
}
\frac{
E\left[1_{[N\neq0]}N^{\rho^+}\right]
}{\rho^++1}.
\end{aligned}
\]
\end{proof}

Thus the implementation shortfall and the marginal cost have the same regular variation order.This means that averaging the marginal price over the execution path does not change the asymptotic price impact exponent. Large trades remain governed by the same tail behaviour as the marginal cost function $F$, while the distribution of the number of informed traders affects the asymptotic scaling factor. In particular, uncertainty about insider participation changes the level of the implementation shortfall through the moments of $N$, but it does not introduce a different asymptotic order.

Furthermore, Theorem~\ref{thm:powerlaw_PriceImpact} identifies the asymptotic distribution of an informed trader's equilibrium order. Since $\Pi^+(F)$ is regularly varying at $\infty$ with index $k_+\rho^+$, where $k_+=\Psi_x^+(M)/(1-\Psi_x^+(M))$, we have
\[
    P(x^*>x)
    =
    P(F^{-1}(V)>x)
    =
    P(V>F(x))
    =
    \Pi^+(F(x))
    =
    x^{k_+\rho^+}s_+(x),
    \qquad x\rightarrow\infty,
\]
where $s_+$ is slowly varying. Thus the equilibrium trading volume of an individual informed trader inherits the same power-law tail implied by the marginal cost function.

We next show that the aggregate order $Y^*=Nx^*+Z$ has the same right-tail index. This result links the tail behaviour of individual informed trading to the tail behaviour of total market order flow.

\begin{corollary}\label{corollary:P(Y^{*})}
Assume the conditions of Theorem~\ref{thm:powerlaw_PriceImpact}. Then the aggregate equilibrium order $Y^*=Nx^*+Z$ satisfies
\[
    P(Y^*>y)=y^{k_+\rho^+}s_Y^+(y),
    \qquad y\rightarrow\infty,
\]
where $s_Y^+$ is slowly varying. Similarly,
\[
    P(Y^*<y)=|y|^{k_-\rho^-}s_Y^-(|y|),
    \qquad y\rightarrow-\infty,
\]
where $s_Y^-$ is slowly varying.
\end{corollary}

\begin{proof}
We prove the right-tail statement. The left-tail statement is analogous. For $y>0$,
\[
\begin{aligned}
P(Y^*>y)
&=
P(Nx^*+Z>y) \\
&=
p(0)P(Z\geq y)
+
\sum_{n\geq1}p(n)
\int_{-\infty}^{\infty}
P\left(x^*>\frac{y-z}{n}\right)q(\sigma,z)\,dz \\
&=
p(0)P(Z\geq y)
+
\sum_{n\geq1}p(n)
\int_{-\infty}^{\infty}
\Pi^+\left(F\left(\frac{y-z}{n}\right)\right)q(\sigma,z)\,dz .
\end{aligned}
\]
Equivalently, after the change of variables $r=(y-z)/n$,
\[
    P(Y^*>y)
    =
    p(0)P(Z\geq y)
    +
    \sum_{n\geq1}p(n)n
    \int_{-\infty}^{\infty}
    \Pi^+(F(r))q(\sigma,y-nr)\,dr .
\]
Let $L(y):=\Pi^+(F(y))$. By Theorem~\ref{thm:powerlaw_PriceImpact}, $L$ is regularly varying at $\infty$ with index $k_+\rho^+<0$. Since the Gaussian tail is exponentially small, $P(Z\geq \alpha y)/L(\alpha)\rightarrow0$ for every fixed $y>0$. Hence, for $y>0$,
\[
\begin{aligned}
\frac{P(Y^*>\alpha y)}{P(Y^*>\alpha)}
&=
\frac{
p(0)P(Z\geq\alpha y)
+
\sum_{n\geq1}p(n)n
\int_{-\infty}^{\infty}
L(r)q(\sigma,\alpha y-nr)\,dr
}{
p(0)P(Z\geq\alpha)
+
\sum_{n\geq1}p(n)n
\int_{-\infty}^{\infty}
L(r)q(\sigma,\alpha-nr)\,dr
} \\
&\sim
\frac{
\sum_{n\geq1}p(n)
\int_{-\infty}^{\infty}
\frac{L(\alpha r)}{L(\alpha)}
q\left(\frac{\sigma}{\alpha n},\frac{y}{n}-r\right)\,dr
}{
\sum_{n\geq1}p(n)
\int_{-\infty}^{\infty}
\frac{L(\alpha r)}{L(\alpha)}
q\left(\frac{\sigma}{\alpha n},\frac{1}{n}-r\right)\,dr
}.
\end{aligned}
\]
The moment condition in Theorem~\ref{thm:powerlaw_PriceImpact} justifies the interchange of the limit and the summation over $n$. Since $q(\sigma/(\alpha n),y/n-r)\,dr$ converges weakly to the point mass at $y/n$, regular variation of $L$ gives
\[
\begin{aligned}
\lim_{\alpha\rightarrow\infty}
\frac{P(Y^*>\alpha y)}{P(Y^*>\alpha)}
&=
\frac{
\sum_{n\geq1}p(n)\left(\frac{y}{n}\right)^{k_+\rho^+}
}{
\sum_{n\geq1}p(n)\left(\frac{1}{n}\right)^{k_+\rho^+}
}
=
y^{k_+\rho^+}.
\end{aligned}
\]
Therefore $P(Y^*>y)$ is regularly varying at $\infty$ with index $k_+\rho^+$. This proves the right-tail statement.
\end{proof}

Recall that $\Psi_x^+(M)<1$ captures a power-type upper endpoint tail of the asset value distribution. The results above show that this endpoint tail behaviour is transmitted through the equilibrium: the marginal cost $F$, the marginal price $h$, the implementation shortfall, the individual informed order $x^*$, and the aggregate order $Y^*$ all exhibit power-law asymptotics. The random number of informed traders affects the constants and the admissible exponents through the distribution of $N$, but it does not destroy the power-law structure of the equilibrium.

In contrast, when \( \Psi_x^+(M) = 1 \), corresponding to a light-tailed asset value distribution, the asymptotic behaviour of price impact is no longer polynomial. In this regime, the marginal cost function grows logarithmically in trade size. This is formalised in the following theorem.
\begin{theorem}\label{thm:logarithmic_PriceImpact}
Assume Assumption~\ref{assumption:Psi} holds, $p(0)<1$, and $-\infty<m<M<\infty$. Let $N$ have probability mass function $p(n)$, $n=0,1,2,\dots$, and let $F$ be a strictly increasing solution of \eqref{eq:F_FixedPointMapping}. In the light endpoint case $\Psi_x^+(M)=\Psi_x^-(m)=1$, the following statements hold:
\begin{enumerate}[label=(\arabic*)]
    \item Suppose that there exist an integer $r_+\geq1$ and a constant $\kappa_+\in(0,\infty)$ such that
    \begin{align}\label{eq:cond_Psi_x=1_plus}
        \lim_{x\rightarrow M}
        \frac{\Psi^+(x)-x}{(M-x)^{r_++1}}
        =
        \frac{1}{\kappa_+}.
    \end{align}
    Assume further that the expectations appearing in \eqref{eq:log_fixedpoint_plus} are finite. Then $\Pi^+(F)$ is regularly varying at $\infty$ with index $-\kappa_+\rho^+$, where $\rho^+>0$ solves
    \begin{align}\label{eq:log_fixedpoint_plus}
        \rho^+
        =
        \frac{1}{\kappa_+}
        +
        \frac{\rho^+}{1-p(0)}
        E\left[
        1_{[N\neq0]}
        \left(\log N+\frac{1}{N}\right)
        \right]
        -
        \rho^+
        \frac{
        E\left[1_{[N\neq0]}N^{\kappa_+\rho^+}\log N\right]
        }{
        E\left[1_{[N\neq0]}N^{\kappa_+\rho^+}\right]
        }.
    \end{align}
    Moreover,
    \[
        M-F(x)
        \sim
        (r_+\rho^+)^{-1/r_+}(\log x)^{-1/r_+},
        \qquad x\rightarrow\infty .
    \]

    \item Suppose that there exist an integer $r_-\geq1$ and a constant $\kappa_-\in(0,\infty)$ such that
    \begin{align}\label{eq:cond_Psi_x=1_minus}
        \lim_{x\rightarrow m}
        \frac{x-\Psi^-(x)}{(x-m)^{r_-+1}}
        =
        \frac{1}{\kappa_-}.
    \end{align}
    Assume further that the expectations appearing in \eqref{eq:log_fixedpoint_minus} are finite. Then $\Pi^-(F)$ is regularly varying at $-\infty$ with index $-\kappa_-\rho^-$, where $\rho^->0$ solves
    \begin{align}\label{eq:log_fixedpoint_minus}
        \rho^-
        =
        \frac{1}{\kappa_-}
        +
        \frac{\rho^-}{1-p(0)}
        E\left[
        1_{[N\neq0]}
        \left(\log N+\frac{1}{N}\right)
        \right]
        -
        \rho^-
        \frac{
        E\left[1_{[N\neq0]}N^{\kappa_-\rho^-}\log N\right]
        }{
        E\left[1_{[N\neq0]}N^{\kappa_-\rho^-}\right]
        }.
    \end{align}
    Moreover,
    \[
        F(x)-m
        \sim
        (r_-\rho^-)^{-1/r_-}(\log |x|)^{-1/r_-},
        \qquad x\rightarrow-\infty .
    \]
\end{enumerate}
\end{theorem}

\begin{proof}
We prove the right tail statement. The left tail is analogous. Set
\[
    G(x):=M-F(x),\qquad x>0.
\]
In the light endpoint case $\Psi_x^+(M)=1$, the power law exponent degenerates to zero, and the relevant scaling is logarithmic. For $\alpha>0$, define
\[
    r_\alpha(x)
    :=
    \frac{F(\alpha x)-F(\alpha)}{G(\alpha)^{r_++1}},
    \qquad x>0.
\]
The compactness and comparison argument follows the proof of Theorem~B.2 of \cite{Cetin_Waelbroeck_2023}. In the present setting, the deterministic number of insiders is replaced by averaging with respect to the distribution of $N$, and the finiteness of the expectations appearing below justifies the corresponding interchange of limits and summation. Hence every subsequential limit of $r_\alpha$ is of the form
\(
    \gamma(x)=\rho^+\log x,
\)
for some $\rho^+>0$.

We first show that $\Pi^+(F)$ is regularly varying. Define
\[
    g(x):=\exp\{G(x)^{-r_+}\}.
\]
Since
\[
\begin{aligned}
\log\frac{g(\alpha x)}{g(\alpha)}
&=
G(\alpha x)^{-r_+}-G(\alpha)^{-r_+}  \\
&=
\frac{F(\alpha x)-F(\alpha)}{G(\alpha)^{r_++1}}
\frac{G(\alpha)}{G(\alpha x)}
\sum_{i=1}^{r_+}
\left(\frac{G(\alpha x)}{G(\alpha)}\right)^i ,
\end{aligned}
\]
and $G(\alpha x)/G(\alpha)\rightarrow1$, we obtain
\[
    \frac{g(\alpha x)}{g(\alpha)}
    \rightarrow
    \exp\{r_+\rho^+\log x\}
    =
    x^{r_+\rho^+}.
\]
Thus $g$ is regularly varying at $\infty$ with index $r_+\rho^+$. By Karamata's theorem,
\begin{align}\label{eq:log_F_derivative_limit}
    \lim_{x\rightarrow\infty}
    \frac{xF'(x)}{G(x)^{r_++1}}
    =
    \rho^+ .
\end{align}
Moreover, the endpoint assumption \eqref{eq:cond_Psi_x=1_plus} gives
\[
    \Psi^+(F(x))-F(x)
    \sim
    \frac{1}{\kappa_+}G(x)^{r_++1}.
\]
Using
\[
    -\frac{\Pi_x^+(x)}{\Pi^+(x)}
    =
    \frac{\Psi_x^+(x)}{\Psi^+(x)-x},
\]
we obtain
\[
\begin{aligned}
\lim_{x\rightarrow\infty}
\frac{\Pi^+(F(x))}
{\int_x^\infty \Pi^+(F(t))\,dt/t}
&=
-\lim_{x\rightarrow\infty}
\frac{x\Pi_x^+(F(x))F'(x)}{\Pi^+(F(x))} \\
&=
\lim_{x\rightarrow\infty}
\frac{x\Psi_x^+(F(x))F'(x)}
{\Psi^+(F(x))-F(x)}
=
\kappa_+\rho^+ .
\end{aligned}
\]
Therefore, by the converse part of Karamata's theorem, $\Pi^+(F)$ is regularly varying at $\infty$ with index $-\kappa_+\rho^+$.

It remains to identify $\rho^+$. The preceding regular variation implies
\[
    \frac{\Pi^+(F(\alpha y))}{\Pi^+(F(\alpha))}
    \rightarrow
    y^{-\kappa_+\rho^+}.
\]
Introduce the tilted probability mass function
\[
    p^1(n;\kappa_+,\rho^+)
    :=
    \frac{p(n)1_{[n\neq0]}n^{\kappa_+\rho^+}}
    {E\left[1_{[N\neq0]}N^{\kappa_+\rho^+}\right]},
    \qquad n\geq1.
\]
By Definition~\ref{def:phi_g}, the concentration of
$q(\sigma/(\alpha n),y-z/n)\,dy$ at $z/n$, and the endpoint condition, we have
\[
    \frac{\phi_F^+(\alpha z)-F(\alpha)}
    {G(\alpha)^{r_++1}}
    \rightarrow
    \frac{1}{\kappa_+}
    +
    E_{p^1}\left[
    \rho^+\log\left(\frac{z}{N}\right)
    \right],
    \qquad z>0.
\]
Substituting this limit into the scaled fixed point equation gives, for $x>0$,
\[
\begin{aligned}
\rho^+\log x
&=
\frac{1}{\kappa_+}
+
\sum_{n\geq1}\frac{p(n)}{1-p(0)}
\left\{
\frac{1}{n}
E_{p^1}\left[
\rho^+\log\left(\frac{nx}{N}\right)
\right]\right.\\
&\hspace{3.4cm}\left.
+
\frac{n-1}{nx}
\int_0^x
E_{p^1}\left[
\rho^+\log\left(\frac{ny}{N}\right)
\right]\,dy
\right\}.
\end{aligned}
\]
Evaluating the elementary integral and cancelling the common $\rho^+\log x$ terms yields
\[
    \rho^+
    =
    \frac{1}{\kappa_+}
    +
    \frac{\rho^+}{1-p(0)}
    E\left[
    1_{[N\neq0]}
    \left(\log N+\frac{1}{N}\right)
    \right]
    -
    \rho^+
    \frac{
    E\left[1_{[N\neq0]}N^{\kappa_+\rho^+}\log N\right]
    }{
    E\left[1_{[N\neq0]}N^{\kappa_+\rho^+}\right]
    }.
\]

Finally, since $g(x)=\exp\{G(x)^{-r_+}\}$ is regularly varying with index $r_+\rho^+$,
\[
    G(x)^{-r_+}
    =
    r_+\rho^+\log x+o(\log x).
\]
Therefore
\[
    M-F(x)=G(x)
    \sim
    (r_+\rho^+)^{-1/r_+}(\log x)^{-1/r_+},
    \qquad x\rightarrow\infty .
\]
\end{proof}

\medskip

Theorem~\ref{thm:logarithmic_PriceImpact} shows that light endpoint tails generate logarithmic, rather than power-law, price impact. When the asset value distribution has a sufficiently thin endpoint tail, increasingly large buy orders reveal information only gradually as the price approaches the upper bound $M$. As a result, the marginal cost $F(x)$ converges to $M$ at the logarithmic rate $(\log x)^{-1/r_+}$, and the corresponding sell-side behaviour is analogous near the lower bound $m$. The uncertainty in the number of informed traders changes the logarithmic coefficient through the distribution of $N$, but it does not change the logarithmic nature of the asymptotic price impact. When the distribution of $N$ degenerates to a deterministic value, the result reduces to the logarithmic price impact obtained in \cite{Cetin_Waelbroeck_2023}.

\section{Numerical Experiments}\label{section:Numerics}

In this section, we study the equilibrium model numerically. We first solve the fixed point equation \eqref{eq:F_FixedPointMapping} for the marginal cost function \(F\), and then recover the equilibrium limit price function \(h^*=\phi_F\) from Definition~\ref{def:phi_g}. This provides the numerical counterpart of the equilibrium characterisation in Section~\ref{section::generalRandom}. We then examine the large-order behaviour of price impact and compare the numerical results with the asymptotic predictions in Section~\ref{section:PriceAsymptotics}. Finally, we study how the equilibrium changes with both the distribution of the number of informed traders and the distribution of the asset value. Some of these numerical experiments fall outside the assumptions of the theoretical results above, and are therefore intended to document robust empirical patterns suggested by the fixed point equation.

\subsection{Setup}\label{subsection:Setup}

We consider several specifications for the asset value \(V\). The first is a bounded power-law family on \((-M,M)\), with density proportional to \((1-|v|/M)^{\alpha-1}\). For this family, the upper endpoint satisfies \(\Psi_x^+(M)=\alpha/(\alpha+1)\in(0,1)\), and the lower endpoint is symmetric.
Hence this specification falls directly under the power-law regime of Theorem~\ref{thm:powerlaw_PriceImpact}. We also consider two unbounded heavy-tailed distributions: the Student-\(t\) distribution with \(\nu\) degrees of freedom, centred at zero and symmetric, and the one-sided Pareto distribution with tail index \(\alpha\), supported on the positive half-line. Finally, we use the Gaussian distribution as an unbounded light-tailed benchmark. Unless otherwise specified, the Gaussian benchmark is \(\mathcal N(0,1)\), and the Student-\(t\) distribution is scaled to have variance one; in particular, when \(\nu=3\), we use scale \(1/\sqrt{3}\). The Student-\(t\) is symmetric about zero like the bounded family, whereas the one-sided Pareto is asymmetric (\(V>0\) almost surely), included to see how the equilibrium responds when the value distribution is not symmetric. The Student-\(t\) and Pareto specifications are not covered by the bounded-support existence result in Theorem~\ref{thm:equilibrium_existence}, nor by the bounded-endpoint asymptotic theory in Theorem~\ref{thm:powerlaw_PriceImpact}. They are included to test whether the fixed point iteration remains stable and whether the power-law pattern remains visible beyond the compact-support setting. The Gaussian case is used to compare the numerical behaviour with the logarithmic type of price impact identified in the light-endpoint regime of
Theorem~\ref{thm:logarithmic_PriceImpact}.

For the number of informed traders \(N\), we use the negative binomial distribution with mean \(\mu\) and dispersion parameter \(r>0\),
\[
    P(N=n)=\frac{\Gamma(n+r)}{n!\,\Gamma(r)}
    \left(\frac{r}{r+\mu}\right)^r
    \left(\frac{\mu}{r+\mu}\right)^n,
    \qquad n=0,1,2,\dots .
\]
This specification includes the geometric distribution when \(r=1\) and approaches the Poisson distribution as \(r\rightarrow\infty\). Since \(\operatorname{Var}(N)=\mu+\mu^2/r\), increasing \(r\) reduces the dispersion of \(N\) while keeping the mean fixed. This allows us to vary the dispersion of \(N\) while controlling the relevant mean in the numerical experiments below.

We solve the equilibrium fixed point equation by iteration, \(F^{(j+1)}=T(F^{(j)}),\) where \(T\) denotes the operator in \eqref{eq:F_FixedPointMapping}. The iteration is stopped when \(\|F^{(j+1)}-F^{(j)}\|_\infty<10^{-8}\). All functions are computed on a symmetric grid, and the Gaussian noise integrals appearing in the fixed point operator are evaluated by numerical quadrature.
Unless otherwise stated, all numerical experiments use noise trading variance
\(\sigma_Z^2=0.25\).

\subsection{Numerical Results}\label{subsection:Results}

We begin with the numerical equilibrium. Figure~\ref{fig:Fh} displays the marginal cost function \(F\) and the limit price in equilibrium \(h^*=\phi_F\), obtained by solving the fixed point equation \eqref{eq:F_FixedPointMapping} and then applying the pricing operator in Definition~\ref{def:phi_g}. The top row plots \(F\), while the bottom row plots \(h\). Each column corresponds to a different specification of \(V\), and the coloured curves vary the dispersion parameter \(r\) of the negative-binomial distribution for \(N\). Throughout, \(E[N\mid N\geq1]=3\). Concretely, for each dispersion \(r\), we choose the unconditional mean \(\mu=E[N]\) so that \(\mu/(1-p_0)=3\), where \(p_0=(1+\mu/r)^{-r}=P(N=0)\). Thus the conditional mean is held fixed while the unconditional mean and the probability of informed trading vary with \(r\).

Across all four asset distributions, the ordering of the curves is stable. Holding \(E[N\mid N\geq1]\) fixed, price impact is steepest when \(r=30\), which is close to the Poisson case, and becomes flatter as \(r\) decreases toward the geometric case \(r=1\). Thus, in these specifications, greater dispersion in the number of informed traders leads to a less steep marginal cost and limit price, and hence to a more liquid book away from the origin. The limit price also exhibits a jump at the origin, equal to the equilibrium bid--ask spread \(h^*(0+)-h^*(0-)\), reflecting the compensation liquidity suppliers require for adverse selection.

\begin{figure}[!htbp]
  \centering
  \includegraphics[width=\textwidth]{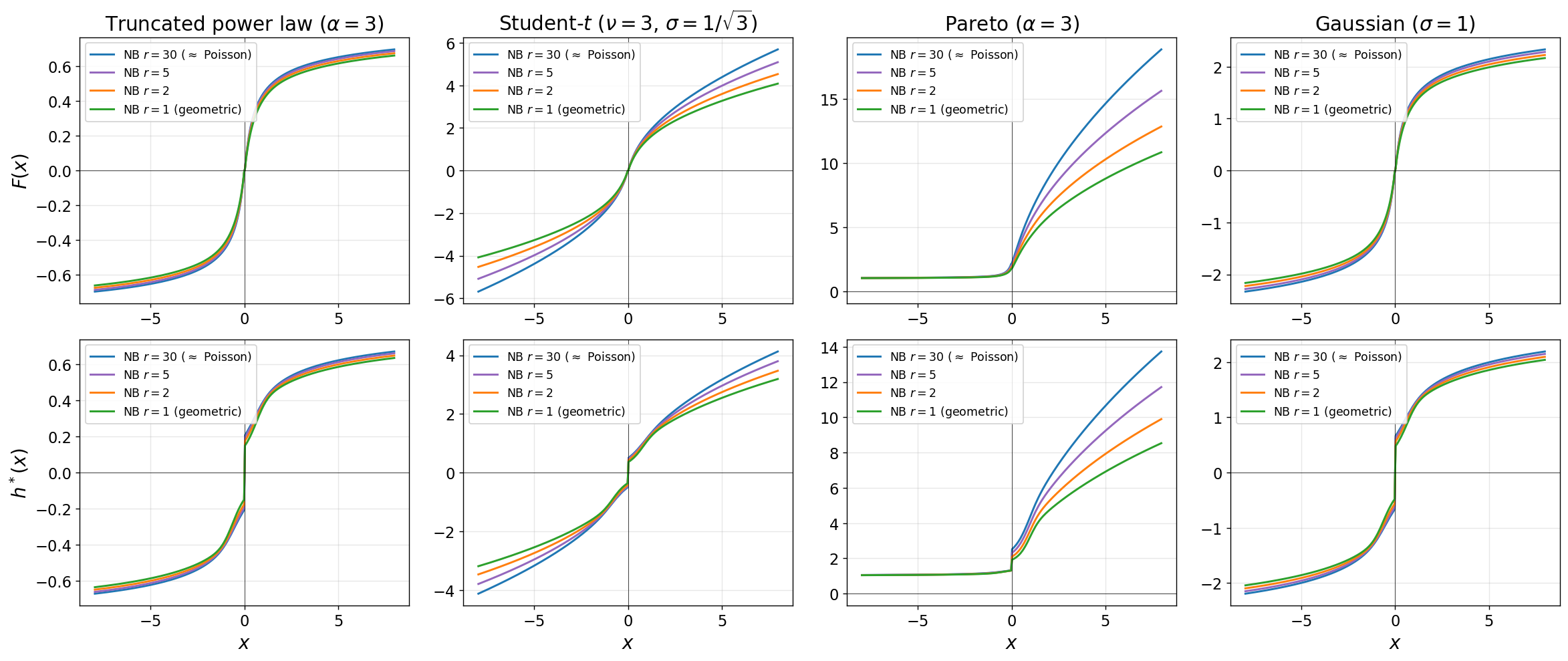}
  \caption{Equilibrium marginal cost \(F\) (top row) and limit price \(h^*=\phi_F\)
  (bottom row) as functions of the order size \(x\). The columns correspond to a
  truncated power law with \(\alpha=3\), Student-\(t\) with \(\nu=3\), Pareto
  with \(\alpha=3\), and Gaussian asset value distribution. The number of informed traders follows a negative-binomial distribution with fixed
  conditional mean \(E[N\mid N\ge1]=3\) and dispersion \(r\in\{1,2,5,30\}\).}
  \label{fig:Fh}
\end{figure}

Figure~\ref{fig:powerlaw_fit} compares the numerical tail behaviour with the asymptotic predictions. For the bounded power-law asset, the fitted log--log slope of \(M-F(x)\) matches the theoretical index \(\rho^+\) from Theorem~\ref{thm:powerlaw_PriceImpact}. The same comparison for \(M-h^*(x)\) in panel~\ref{fig:powerlaw_fit_h} confirms the theoretical implication that \(M-h^*\) and \(M-F\) are regularly varying with the same index; see \eqref{eq:h_F_ratio}. For the unbounded Student-\(t\) and Pareto specifications, which are outside the bounded-support theory, both \(F\) and \(h^*\) display clear power-law growth over the fitted range. The Gaussian case is instead approximately linear in \(\sqrt{\ln x}\), consistent with logarithmic price impact.

\begin{figure}[!htbp]
  \centering
  \begin{subfigure}{\textwidth}
    \centering
    \includegraphics[width=\textwidth]{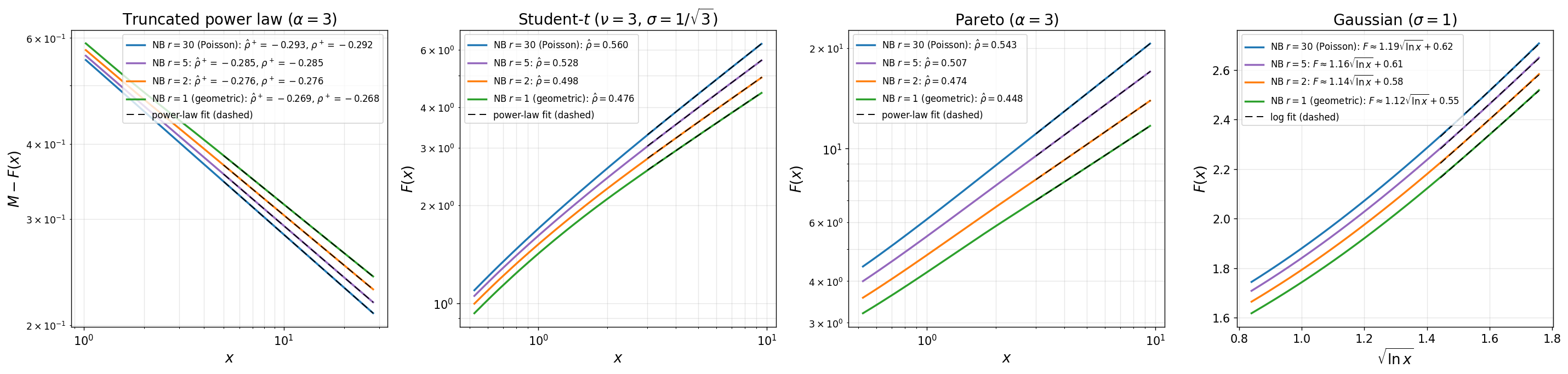}
    \caption{Marginal cost \(F\).}
    \label{fig:powerlaw_fit_F}
  \end{subfigure}

  \vspace{0.6em}

  \begin{subfigure}{\textwidth}
    \centering
    \includegraphics[width=\textwidth]{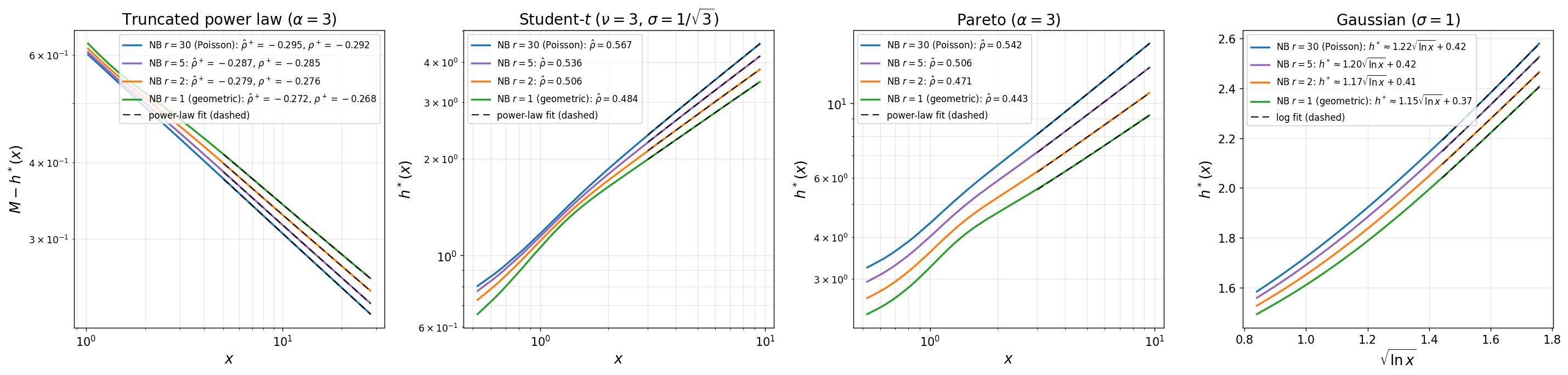}
    \caption{Limit price \(h^*\).}
    \label{fig:powerlaw_fit_h}
  \end{subfigure}

  \caption{Tail fits for the marginal cost \(F\) and limit price \(h^*=\phi_F\). The conditional mean is fixed at \(E[N\mid N\ge1]=3\), with dispersion \(r\in\{1,2,5,30\}\). Solid lines are numerical equilibrium tails; black dashed lines are fitted power-law or logarithmic tails.}
  \label{fig:powerlaw_fit}
\end{figure}

Figure~\ref{fig:bounded_index_statics} studies comparative statics of the regular variation index \(\rho^+\) for the bounded power law asset, where \(\rho^+\in(-1,0)\) governs the tail behaviour of \(M-F(x)\) and is characterised by the fixed point equation \eqref{eq:PowerLaw_fixedpoint_plus}.
The left panel varies the endpoint parameter \(\Psi_x^+(M)\). Since \(\rho^+<0\), Proposition~\ref{prop:tail_index_price_impact_exponent} predicts that \(|\rho^+|\) decreases as \(\Psi_x^+(M)\) increases. The numerical curves confirm this monotonicity over the plotted range \(\Psi_x^+(M)\ge1/2\), corresponding to \(\alpha\ge1\), where the bounded density does not diverge at the support endpoints.

The right panel varies \(E[N\lvert N\ge1]\) at fixed \(\Psi_x^+(M)=0.75\). Along each fixed \(r\) negative binomial subfamily, \(|\rho^+|\) decreases as this conditional mean increases, consistent with the deterministic benchmark. Because changing the mean parameter also changes higher conditional features of \(N\lvert N\ge1\), this comparison should be read as a negative binomial family comparative static rather than a pure conditional mean effect.
The coloured curves show that dispersion in \(N\lvert N\ge1\) also affects the tail index. The ordering across \(r\) can vary with the asset tail parameter \(\Psi_x^+(M)\). For the specifications shown here, the more dispersed negative binomial laws produce smaller values of \(|\rho^+|\). We therefore view this as a numerical pattern for the present parametrisation, rather than as a general comparative statics result.

A natural interpretation is that greater dispersion in the conditional distribution of \(N\), which we use here as a measure of uncertainty about the realized number of informed traders, weakens the informational content of a given aggregate order. When liquidity suppliers observe a large order, they must distinguish whether it reflects a large individual informed trade, a larger number of informed traders, or noise demand. Greater dispersion in \(N\mid N\geq 1\) makes this inference less direct. As a result, the same order flow may carry weaker adverse selection content, and liquidity suppliers adjust marginal prices less aggressively, leading to a flatter equilibrium book in the specifications considered here.

\begin{figure}[!htbp]
  \centering
  \includegraphics[width=0.8\textwidth]{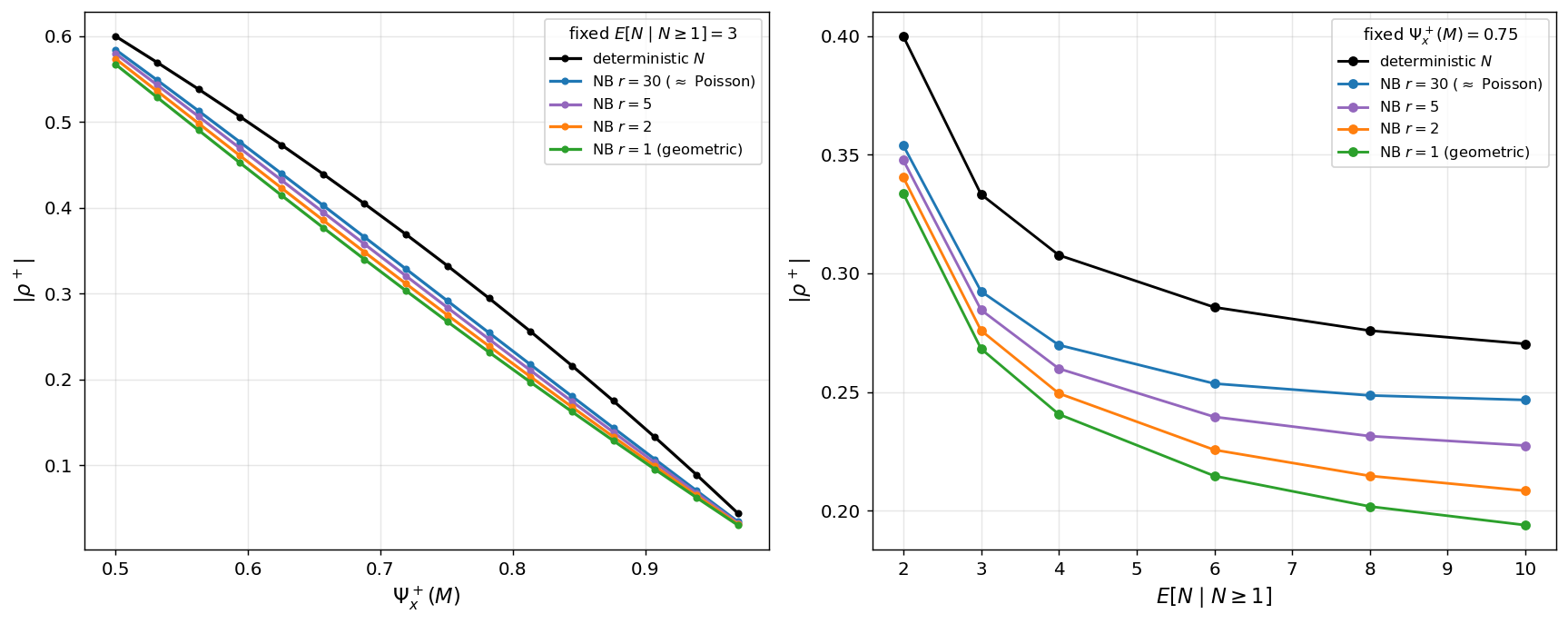}
  \caption{Comparative statics of the market impact index magnitude \(|\rho^+|\) for the bounded power-law asset, where
  \(\Psi_x^+(M)=\alpha/(\alpha+1)\). The left panel varies \(\Psi_x^+(M)\) at fixed \(E[N\mid N\ge1]=3\), while the right panel varies \(E[N\mid N\ge1]\) at fixed \(\Psi_x^+(M)=0.75\). The black curve is the deterministic insider-count benchmark; coloured curves correspond to negative-binomial specifications with different dispersion parameters \(r\).}
  \label{fig:bounded_index_statics}
\end{figure}

Figure~\ref{fig:spread} returns to the bid--ask spread and considers the two-point informed-trader count \(P(N=0)=1-p\) and \(P(N=m)=p\), as in \eqref{eq:twopoint_spread}. For each asset distribution, we solve the fixed point for \(F\) and evaluate \(h^*(0^+)-h^*(0^-)\). The left panel varies the presence probability \(p\) at fixed \(m=3\), while the right panel varies the active number \(m\) at fixed \(p=0.5\). All four asset distributions are scaled to \(\operatorname{Var}(V)=1\). In all four specifications, the spread increases in both \(p\) and \(m\).

The asset distribution also affects the level and shape of the spread. These features are not determined by a single summary statistic, such as tail index, dispersion, or mean, but by the full law of \(V\). This final exercise should be therefore read as a numerical comparative static for the two-point specification of \(N\), rather than as a general monotonicity result for arbitrary informed-trader count distributions.

\begin{figure}[!htbp]
  \centering
  \includegraphics[width=0.8\textwidth]{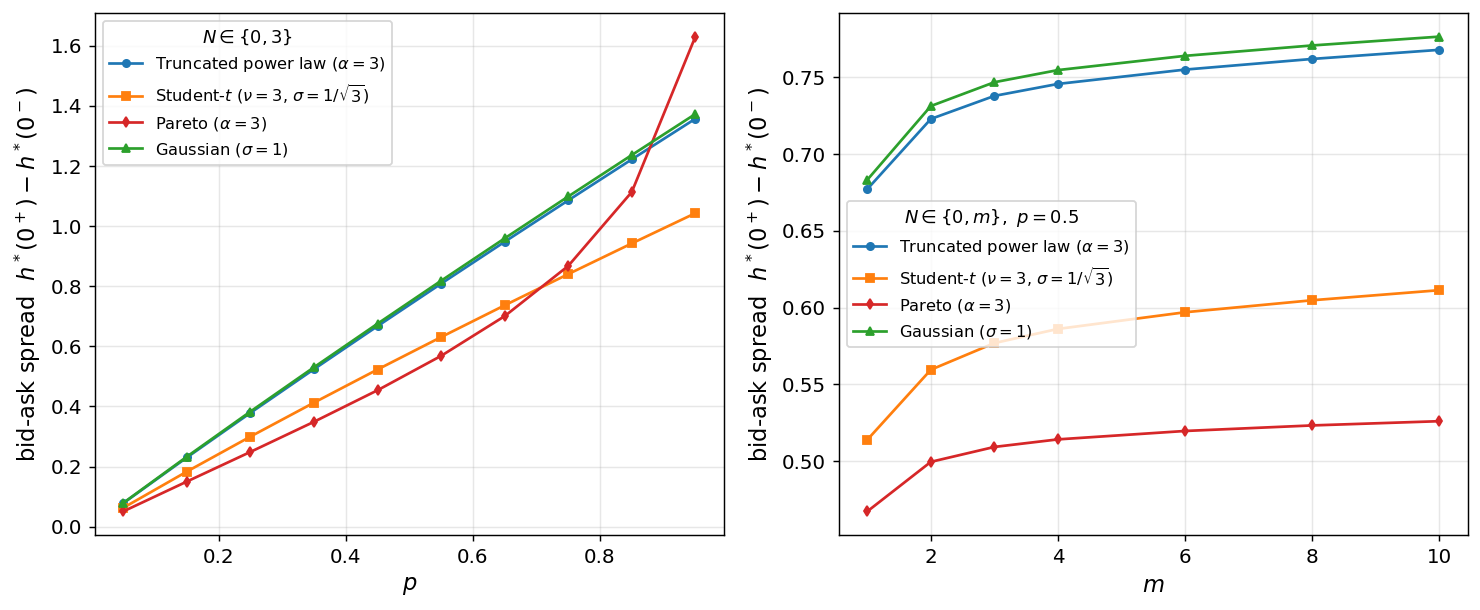}
  \caption{Equilibrium bid--ask spread \(h^*(0^+)-h^*(0^-)\) for a two-point
  informed-trader count \(N\in\{0,m\}\), with \(P(N=m)=p\). The four asset value
  distributions are scaled to \(\operatorname{Var}(V)=1\): truncated power law
  (\(\alpha=3\)), Student-\(t\) (\(\nu=3\)), one-sided Pareto (\(\alpha=3\)),
  and Gaussian. Left: spread as a function of the presence probability \(p\),
  with \(m=3\). Right: spread as a function of the active number \(m\), with
  \(p=0.5\).}
  \label{fig:spread}
\end{figure}

Finally, the insider's profit mentioned in Example \ref{example:binary} is examined numerically in Appendix~\ref{appendix:profit}.

\section{Conclusion}\label{section:Conclusion}

This paper studies a one period limit order market in which the number of informed traders is random. We characterize equilibrium through a fixed point equation for the marginal cost function \(F\), and recover the equilibrium limit order book \(h^*=\phi_F\) from the associated pricing operator. This framework nests the deterministic insider benchmark while allowing liquidity suppliers to price both adverse selection and uncertainty about the number of informed traders.

We also derive large order asymptotics for equilibrium price impact. In the bounded support case, the tail behaviour of \(M-F(x)\) is governed by a regular variation index determined jointly by the endpoint behaviour of the asset value distribution and the distribution of the number of informed traders. In the light endpoint regime, the model instead generates logarithmic price impact. These results show that the distribution of informed trading participation affects not only the level of liquidity but also the asymptotic shape of the limit order book. In particular, the effect of uncertainty about the informed trader count is not summarized by the expected number of informed traders alone.

The numerical experiments complement the theoretical results by solving the fixed point directly across several asset value and informed trader count specifications. They confirm the predicted power law and logarithmic tail behaviour in the cases covered by the theory, and document robust patterns beyond the compact support setting. The numerical results further suggest that, at a fixed conditional mean, greater dispersion in the number of informed traders can flatten the equilibrium book in the specifications considered. This pattern is consistent with a signal extraction interpretation: when the realized number of informed traders is more uncertain, aggregate order flow is less directly informative about the asset value, and liquidity suppliers may adjust marginal prices less aggressively. We view this dispersion effect as a numerical finding rather than a general comparative statics theorem.

\clearpage
\appendix

\section{Karamata's Theorems}\label{appendix:Karamata}

The following direct and converse forms of Karamata's theorem are used in the
derivation of the price-impact asymptotics; see, for example,
\cite[Theorems~1.5.11 and~1.6.1]{Bingham_Goldie_Teugels_1987}. They describe
the asymptotic behaviour of regularly varying functions when integrated against
power functions.

\begin{theorem}[Karamata's theorem: direct half]\label{thm:KaramataDirect}
Let \(f\) be regularly varying at infinity with index \(\rho\), and suppose that
\(f\) is locally bounded on \([X,\infty)\). Then:
\begin{enumerate}[label=(\arabic*)]
    \item for any \(\sigma\geq-(\rho+1)\),
    \[
        \frac{x^{\sigma+1}f(x)}
        {\int_X^x t^\sigma f(t)\,dt}
        \longrightarrow
        \sigma+\rho+1,
        \qquad x\rightarrow\infty;
    \]
    \item for any \(\sigma<-(\rho+1)\),
    \[
        \frac{x^{\sigma+1}f(x)}
        {\int_x^\infty t^\sigma f(t)\,dt}
        \longrightarrow
        -(\sigma+\rho+1),
        \qquad x\rightarrow\infty .
    \]
\end{enumerate}
\end{theorem}

\begin{theorem}[Karamata's theorem: converse half]\label{thm:KaramataConverse}
Let \(f\) be positive and locally integrable on \([X,\infty)\). Then:
\begin{enumerate}[label=(\arabic*)]
    \item if, for some \(\sigma>-(\rho+1)\),
    \[
        \frac{x^{\sigma+1}f(x)}
        {\int_X^x t^\sigma f(t)\,dt}
        \longrightarrow
        \sigma+\rho+1,
        \qquad x\rightarrow\infty,
    \]
    then \(f\) is regularly varying at infinity with index \(\rho\);

    \item if, for some \(\sigma<-(\rho+1)\),
    \[
        \frac{x^{\sigma+1}f(x)}
        {\int_x^\infty t^\sigma f(t)\,dt}
        \longrightarrow
        -(\sigma+\rho+1),
        \qquad x\rightarrow\infty,
    \]
    then \(f\) is regularly varying at infinity with index \(\rho\).
\end{enumerate}
\end{theorem}

\section{Insider Profit under an Uncertain Count}\label{appendix:profit}

Example~\ref{example:binary} shows analytically, for a binary asset value and a
possibly absent monopolistic insider, that uncertainty about the insider raises the insider's expected profit. We complement this result numerically for continuous value distributions and a random number of insiders.

An insider who observes \(V\) trades the optimal quantity \(x^*=F^{-1}(V)\), so that \(F(x^*)=V\). 
The conditional expected profit is
\begin{equation}\label{eq:perinsider_profit}
  \Pi(V)
  :=
  x^*\big(F(x^*)-IS(x^*)\big)
  =
  E\!\left[
    \int_0^{x^*}\big(V-h^*(Z+Ny)\big)\,dy
    \,\Big|\, N\ge1
  \right],
  \quad x^*=F^{-1}(V).
\end{equation}
The last expression is the expected gain from executing the position \(x^*\)
against the equilibrium limit order book; its equality with \(x^*(F(x^*)-IS(x^*))\) follows from the definition of implementation shortfall.
In the computed equilibrium, \(\Pi(V)\) is positive away from \(V=0\).

\begin{figure}[!htbp]
  \centering
  \includegraphics[width=\textwidth]{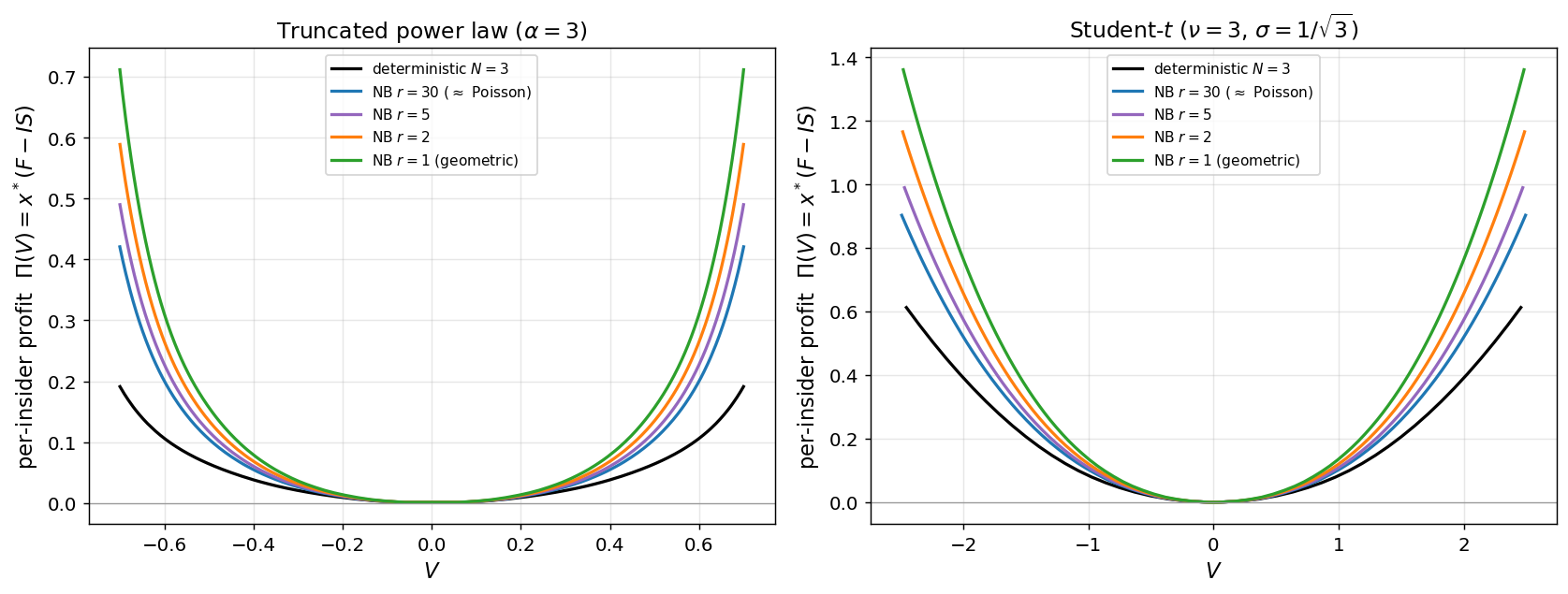}
  \caption{Conditional expected profit \(\Pi(V)\) in
  \eqref{eq:perinsider_profit} for the negative-binomial dispersion ladder at
  fixed \(E[N\mid N\ge1]=3\). Left: truncated power law with \(\alpha=3\).
  Right: Student-\(t\) with \(\nu=3\) and scale \(1/\sqrt3\).}
  \label{fig:perinsider_profit}
\end{figure}

Figure~\ref{fig:perinsider_profit} reports \(\Pi(V)\). The profit is positive
throughout and increases with \(|V|\). Across the dispersion ladder, with
\(E[N\mid N\ge1]=3\) fixed, per-insider profit is lowest in the least dispersed
case and increases as the informed-trader count becomes more dispersed, for both
the bounded and Student-\(t\) value distributions. This mirrors the impact-flattening pattern observed in the numerical experiments: a more dispersed informed-trader count is associated with a flatter marginal cost and limit price, allowing an insider with a given signal to trade a larger position at a lower average execution cost. 
Thus, in these numerical specifications, greater uncertainty about the number of informed traders raises per-insider expected profit, extending the conclusion of Example~\ref{example:binary} beyond the binary-value monopolistic benchmark.

\bibliographystyle{apalike}
\bibliography{biblio.bib}

\end{document}

%% file: 00_preamble.tex
\usepackage[english]{babel}
\usepackage[utf8]{inputenc}
\usepackage[round]{natbib} 

\usepackage[nohead]{geometry} 
\usepackage{setspace} 

\usepackage{hyperref} 
\usepackage{url} 

\usepackage{amsmath} 
\usepackage{amssymb} 
\usepackage{bbm}
\usepackage{amsfonts} 
\usepackage{latexsym} 
\usepackage{amsthm} 
\usepackage{mathrsfs} 
\usepackage{mathtools} 

\usepackage{booktabs} 
\usepackage{multirow} 
\usepackage{multicol} 
\usepackage{enumitem} 

\usepackage{graphics} 
\usepackage{subcaption}  
\usepackage{caption} 
\usepackage{comment} 
\usepackage{accents}
\usepackage{subcaption} 
\usepackage{indentfirst} 
\usepackage{neuralnetwork} 
\usepackage{tikz}
\usepackage{etoolbox} 
\usepackage{listofitems} 
\usepackage{rotating}

\usepackage{authblk} 

\hypersetup{
	colorlinks = {true}, 
	linkcolor = {red}, 
	urlcolor = {black}, 
	citecolor = {blue}
}

\usepackage[commentsnumbered, ruled, vlined]{algorithm2e}

\SetKwInput{KwData}{Data}
\SetKwInput{KwInput}{Input}
\SetKwInput{KwOutput}{Output}
\SetKwInput{KwInit}{Initialize}

\usepackage{lipsum} 
\usepackage{xargs} 
\usepackage[colorinlistoftodos, prependcaption, textsize = small]{todonotes}

\counterwithout{equation}{section}
\newtheorem{example}{Example}[section]
\newtheorem{remark}{Remark}[section]

\input{./def/def_letters.tex}

\numberwithin{equation}{section}
\theoremstyle{plain}

\newtheorem{theorem}{Theorem}[section]
\newtheorem{definition}[theorem]{Definition}
\newtheorem{lemma}[theorem]{Lemma}
\newtheorem{assumption}[theorem]{Assumption}
\newtheorem{proposition}[theorem]{Proposition}
\newtheorem{corollary}[theorem]{Corollary}
\theoremstyle{remark}


%% file: def/def_letters.tex
\def\b0{\boldsymbol{0}}

